\documentclass[fleqn,11pt,letterpaper]{article}
\usepackage{amsfonts}
\usepackage{amssymb}
\usepackage{amsmath}
\usepackage{amsthm}
\usepackage{enumerate}
\usepackage{multirow}
\usepackage{float}
\usepackage{graphicx}
\usepackage{epstopdf}
\usepackage{color}
\usepackage[margin=2.57cm]{geometry}
\definecolor{darkred}{rgb}{0.5,0.2,0.2}
\usepackage[pdftitle={New Approach to CO$_2$ AF},pdfauthor={Mikkel Bennedsen, Eric Hillebrand, Siem Jan Koopman},colorlinks=TRUE,allcolors=darkred]{hyperref}
\usepackage{booktabs}
\usepackage{pdflscape}
\usepackage[round]{natbib}
\usepackage{lineno}

\usepackage{tabularx}

\usepackage{float}
\floatstyle{plaintop}
\restylefloat{table}

\restylefloat{table}

\theoremstyle{plain}
\newtheorem{proposition}{Proposition}
\newtheorem{remark}{Remark}

\newcommand{\E}{\mathbb{E}}

\def\bi{\begin{itemize}}
\def\ei{\end{itemize}}

\allowdisplaybreaks
\graphicspath{{Figures/}}

\newif\ifi

\title{A Regression-Based Approach to the CO$_2$ Airborne Fraction: \\ Enhancing Statistical Precision and Tackling Zero Emissions}

\author{Mikkel Bennedsen, Eric Hillebrand, Siem Jan Koopman}

\begin{document}
\maketitle

\begin{abstract}
The global fraction of anthropogenically emitted carbon dioxide (CO$_2$) that stays in the atmosphere, the CO$_2$ airborne fraction, has been fluctuating around a constant value over the period 1959 to 2022. The consensus estimate of the airborne fraction is around $44\%$; the remaining $56\%$ is absorbed by the oceanic and terrestrials biospheres. 
In this study, we show that the 
conventional estimator 
of the airborne fraction, based on a ratio of changes in atmospheric CO$_2$ concentrations and CO$_2$ emissions, suffers from a number of statistical deficiencies, such as non-existence of moments and a non-Gaussian limiting distribution. 
We propose an alternative regression-based estimator of the airborne fraction that does not suffer from these deficiencies.
We show that the regression-based estimator 
has a Gaussian limiting distribution and reduces estimation
uncertainty substantially.  Our empirical analysis leads to an estimate of
the airborne fraction over 1959--2022 of $47.0\%$ ($\pm 1.1\%$; $1 \sigma$), implying a higher, and better constrained, estimate than the current consensus. 
Using climate model output, we show that a regression-based approach provides sensible estimates of the airborne fraction, also in future scenarios where emissions are at or near zero.
\end{abstract}

\newpage

\clearpage

The amount of anthropogenically emitted carbon dioxide (CO$_2$) that stays in the atmosphere, the so-called \emph{airborne fraction} (AF), is an important quantity for the study of CO$_2$ absorption in the carbon-cycle of the Earth system \citep[][]{BK1973,Siegenthaler1978,Gloor2010}.  
In the literature it has been investigated and debated whether the AF has increased, decreased, or remained constant over the period from 1959 to today, during which atmospheric measurements of CO$_2$ concentrations have been available.
Earlier studies found evidence of an increasing AF \citep[e.g.][]{Canadell2007,Raupach2008,LeQuere2009}, even though measurement and estimation uncertainty make these findings statistically dubious  \citep[e.g.][]{Knorr2009,Ballantyne2015}. Later studies suggest that the AF has remained constant around $44\%$, and this has become the consensus view  \citep{Raupach2014,BHK2019a,Canadell2021,BHK2023a}.  \cite{Raupach2013} shows that the AF is given by a constant in a system where emissions follow an exponential trajectory and the sink uptake is linear in atmospheric CO$_2$ concentrations. \cite{BHK2023b} formalize such a system statistically, also allowing for linear growth of emissions on the more recent sample, and report a point estimate of the AF of $0.44$.

Previous studies have analyzed the AF as the ratio of yearly changes in
atmospheric CO$_2$ concentrations ($G_t$, numerator) and anthropogenic CO$_2$ emissions
($E_t$, denominator) \citep[e.g.][]{Canadell2007,Raupach2008,LeQuere2009,Knorr2009,Ballantyne2012,Raupach2014,Ballantyne2015,keenan2016recent,BHK2019a,vMa2023,Pressburger2023}.
 An alternative to this approach is to consider the 
\emph{cumulative airborne fraction} \citep[CAF; e.g., ][]{Jones2013,Jones2016,Liddicoat2021},  but since this approach  is less commonly used in the literature and is less amenable to statistical analysis, we  only briefly address the CAF in our empirical applications and omit it from our theoretical discussions.    
Instead, we follow the main body of the literature and adopt the  conventional estimator of the AF which is
defined as the sample mean of the yearly ratio $G_t/E_t$.
In this paper, we show that this estimator suffers
from a number of statistical deficiencies due to its definition as the ratio of two stochastic processes. One problem is that the denominator in the ratio, $E_t$, may have a positive probability density at zero, implying that the estimator does not possess any moments.
This lack of moments implies, for instance, that the mean and variance of the estimator do not exist. Although this issue may be of limited concern during the period 1959--2022, it becomes important in future scenarios where CO$_2$ emissions decrease, such as those consistent with ``net-zero'' CO$_2$ emissions, a committed goal of the international community \citep[][]{Rockstrom2017,Riahi2022}.
For example, issues with the conventional analysis of the AF when emissions decrease are
recently encountered in \cite{Pressburger2023},
where the future AF implied by output from a climate model is studied.
Another problem is that the time series of yearly changes in atmospheric concentrations $G_t$ and emissions $E_t$ exhibit a trend,
which implies that a central limit theorem does not hold for the estimator.
Therefore, its limiting distribution may be non-Gaussian,
unless a separate assumption of Gaussianity is imposed. This implies that confidence intervals and $p$-values for test statistics based on the Gaussian distribution
may not be valid for the conventional ratio-based estimator. 

We show that the time series of yearly
changes in atmospheric CO$_2$ concentrations, $G_t$, cointegrates with
the time series of anthropogenic CO$_2$ emissions, $E_t$,
over the period 1959--2022.
Cointegration means that, even though
$G_t$ and $E_t$ are  non-stationary time series individually,
the linear regression equation $G_t = \alpha E_t + u_t$, with constant
parameter $\alpha$, has a stationary error $u_t$. In other words, the non-stationary processes
$G_t$ and $E_t$ are cointegrated when
a linear combination exists that yields a stationary process. 
The parameter $\alpha$ denotes the fraction of emissions that is
added to the atmosphere in year $t$, i.e., in the linear regression, $\alpha$ represents the AF. We show the constancy of $\alpha$ statistically on the historical sample. In future scenarios where emissions decrease substantially, $\alpha$ will be time-varying, however.

We show that estimating  $\alpha$ using an ordinary
least squares method applied to the linear regression equation $G_t = \alpha E_t + u_t$  yields an AF 
estimator with superior statistical properties
compared to the conventional AF estimator based on the ratio $G_t \, / \, E_t$.
This regression-based estimator does not suffer from the defects of the ratio-based estimator described above, and, under mild assumptions, converges to the true value $\alpha$ at the fast rate of $T^{3/2}$,
where $T$ is the sample size.
A central limit theorem applies to the regression-based estimator, i.e. the estimator has a
Gaussian distribution as the sample size increases to infinity for a wide range of distributions
that can be assumed for the regression error $u_t$. Conventional confidence intervals and $p$-values for various test statistics based on the Gaussian distribution are therefore valid for the regression-based estimator. We generalize the regression approach such that $\alpha$ is time-varying for the study of future net-zero emission scenarios.

We apply the ratio-based and the regression-based estimators of the AF to yearly data from the
Global Carbon Project \citep[][]{GCB2023} over the period 1959--2022,
and we find that the regression-based  estimator improves precision, as measured
by the standard error of the estimator, by approximately $11\%$ compared
to the    ratio-based  estimator.
Our estimate of the AF over the period 1959--2022 is $44.8\%$ with an
associated standard error of $1.4\%$. Hence, this estimate is
not significantly different from the consensus estimate of $44\%$.
To further reduce estimation uncertainty, we also consider the
inclusion of additional data
on the El Ni\~no-Southern Oscillation (ENSO) and volcanic activity as covariates 
in the analysis \citep{Raupach2008}.  
In this case, the regression-based estimator improves precision by
approximately $16\%$, compared to the    ratio-based  estimator of the AF.
Our best estimate of the AF over the period 1959--2022 is $47.0\%$
with an associated standard error of $1.1\%$, which leads to a $95\%$
confidence interval of $[44.9\%, 49.0\%]$ for the AF.
The consensus estimate of $44\%$ \citep[][p. 676]{Canadell2021}
falls just outside this confidence interval.
Our empirical results thus imply a higher, and better constrained, AF than
the current consensus.  The higher value of the estimate is mainly driven by the inclusion of ENSO and volcanic activity; the tighter confidence interval is mainly driven by the regression-based estimation approach.

We find that the AF has  remained approximately constant over the historical period 1959--2022,
corroborating earlier results \citep[][]{Raupach2014,BHK2019a,Canadell2021}. 
Future trajectories of emissions will likely result in a non-constant AF, however. This may  be either because of emissions trajectories  departing from exponential growth or from changing dynamics in the carbon sinks, arising from e.g. saturation \citep[][]{LeQuere2007,Canadell2007b} or climate feedback effects \cite[][]{Friedlingstein2015}, resulting in a non-linear relationship between sink activity and atmospheric concentrations. Climate models have shown that the AF tends to increase in future high-emission scenarios and decrease in low-emission scenarios \citep[][]{Jones2013}. 
Using output from the reduced-complexity climate model MAGICC \citep[][]{MAGICC},
we illustrate the challenges in applying the ratio-based estimator of the AF to low-emission scenarios,
and we show that a regression-based approach alleviates these difficulties.
Furthermore, we show that generalizing the regression-based estimator to a time-varying AF can be a powerful tool in analyzing the dynamics of the AF in low-emission scenarios output from climate models, such as scenarios compatible with the Paris Agreement.


\section*{Atmospheric changes, emissions, and cointegration}\label{sec:coint}
Figures \ref{fig:data}a) and \ref{fig:data}b) show yearly changes
in atmospheric CO$_2$ ($G_t$) and yearly CO$_2$ emissions from
anthropogenic sources ($E_t$), respectively. The black line in 
Figure \ref{fig:data}c) shows the ratio of these two variables, $G_t/E_t$.
Data are obtained from the Global Carbon Project and cover
the period $1959$--$2022$; 
in the empirical section below we provide
further details on the data.
The most conspicuous feature of the two data series $G_t$ and $E_t$ is that both time series exhibit  upwards trends. 
Trending behavior is indicative of time series being non-stationary.
A simple least-squares statistical analysis of the bivariate system $(G_t,E_t)$, where the non-stationarity of a time series is not accounted for, yields invalid inference and should be avoided \citep{Granger1974}. However, the notion of cointegration  \citep[see, for example, chapter 19 in][ for a textbook treatment]{hamilton1994} allows us to keep working with the trending time series
$G_t$ and $E_t$ while still obtaining
valid statistical inference. Cointegration methods have been applied in earlier climate studies  \citep[e.g.][]{KS2002,JohansenJoC2012}. 
Informally, two time series are cointegrated if they share a common trend. Formally, the time series $G_t$ and $E_t$ are said to be cointegrated when
both $G_t$ and $E_t$ are non-stationary, and
the error term $u_t$ in the regression equation $G_t = \alpha E_t + u_t$
is stationary.
We adopt the Dickey-Fuller test \citep{dickey1979distribution} to
determine whether a time series is non-stationary.
The test statistic is for the null hypothesis of a unit root, that is, of
having a unit value for the autoregressive dependence of $G_t$
on its lagged value $G_{t-1}$ (and $E_t$ on $E_{t-1}$). The results from this test strongly suggest that both time series are non-stationary (Supplementary Information \ref{app:unitroot}). 
The null hypothesis of no cointegration (that is, $u_t$ is non-stationary)
can be tested formally using the Engle-Granger test \citep{EG1987}, which is a Dickey-Fuller test on the residuals in the regression $G_t = \alpha E_t + u_t$, adjusted for the fact that these residuals are not observed but must be estimated. The null hypothesis of a unit root in $u_t$ is firmly rejected (Supplementary Information \ref{app:unitroot}), and we can therefore conclude that the two yearly time series $G_t$ and $E_t$ are cointegrated.

The cointegration analysis in \ref{app:unitroot} supports the hypothesis that the AF parameter $\alpha$ is constant during the period studied here (1959--2022). 
If the parameter $\alpha$ was changing in a specific direction, this would introduce a trend in the residuals $u_t$.  The result from the Engle-Granger test shows that a trend is not present. This is confirmed graphically by the blue line in Figure \ref{fig:data}f) and is in line with recent studies
\citep{Raupach2014,BHK2019a,Canadell2021,BHK2023a}.
A Jarque-Bera test \citep[][]{JarqueBera1987} for normality of the estimated residuals for $u_t$
results in a $p$-value of $24\%$, implying that we cannot reject the null of $u_t$ having a Gaussian distribution.

\begin{figure}[t!]
    \centering
    \caption{\emph{a): Atmospheric concentration changes ($G_t$) and b): Emissions ($E_t$) data used in the study over the period $1959$--$2022$.
    c): Data  ($G_t/E_t$) in black, fit of \eqref{eq:AForig} in blue (solid), fit of \eqref{eq:AForig2} in red (dashed), $95\%$ confidence bands (shaded), ``intercept'' denotes the ratio-based estimate of
    the AF $\alpha$. d): Ratio-based estimated residuals $\hat u_t$ from \eqref{eq:AForig} in blue (solid) and from \eqref{eq:AForig2} in red (dashed).
    e): Data  ($G_t$) in circles, fit of \eqref{eq:AFnew} in blue (solid), fit of \eqref{eq:AFnew2} in red (dashed), $95\%$ confidence bands (shaded), ``slope''
    denotes the regression-based estimate of the AF $\alpha$. f) Regression-based estimated residuals $\hat u_t$ from \eqref{eq:AFnew} in blue (solid) and from \eqref{eq:AFnew2} in red (dashed).}}
    \includegraphics[width=0.99\textwidth]{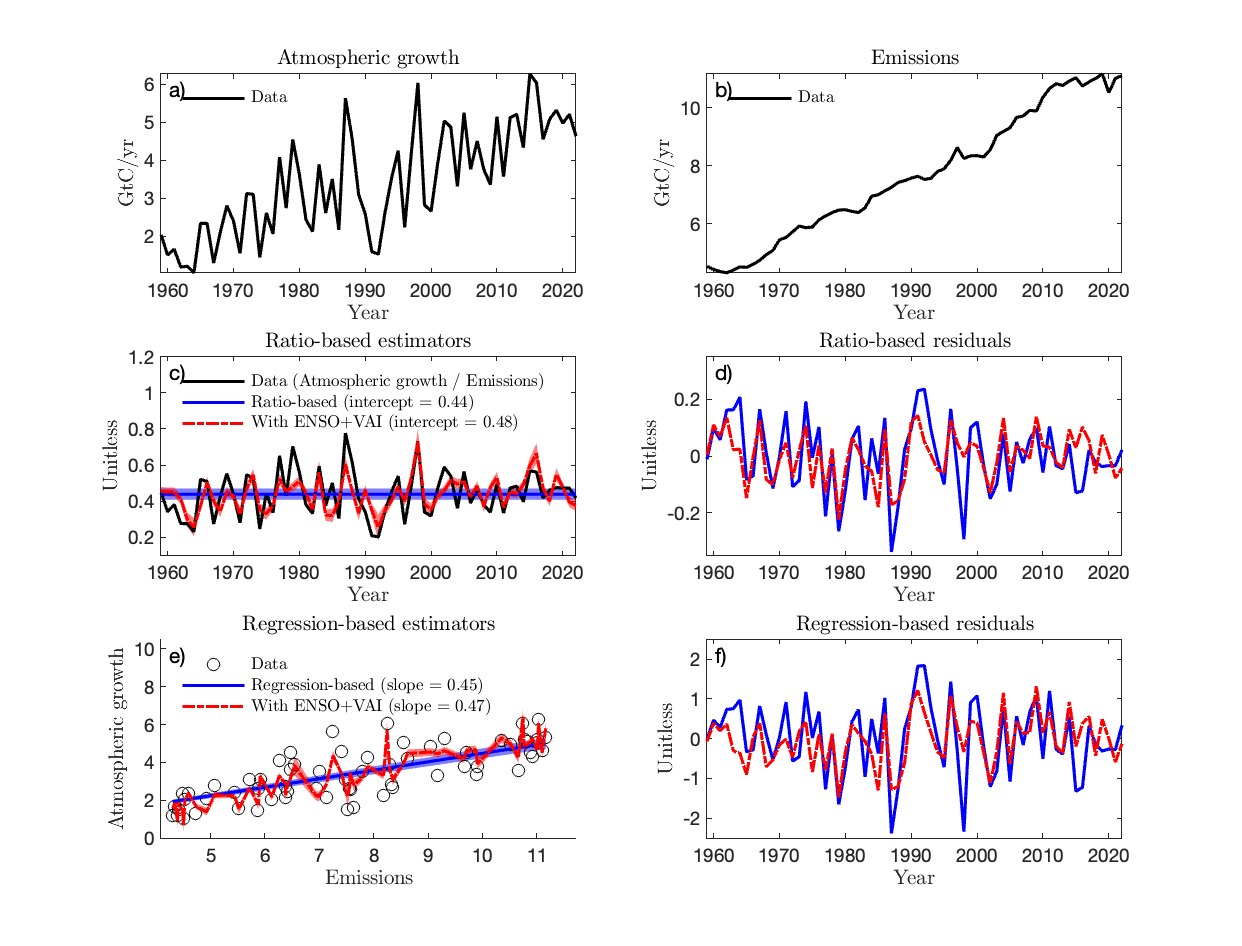}
    \label{fig:data}
\end{figure}

\section*{A regression-based estimator of the airborne fraction}\label{sec:theory}
A ratio-based approach to estimating the AF takes its departure in 
the statistical model given by
\begin{align}\label{eq:AForig}
\frac{G_t}{E_t} = \alpha + u^{(1)}_{t},
\end{align}
where $G_t$ are the yearly changes in atmospheric concentrations of CO$_2$, $E_t$ are yearly CO$_2$ emissions, the constant parameter $\alpha$ is the AF, and
$u^{(1)}_{t}$ is the disturbance modelled as a zero-mean error process,
for $t = 1, 2, \ldots, T$, with $T$ denoting the number of yearly observations in the sample.
The disturbance $u^{(1)}_{t}$ captures deviations of the data ${G_t}/{E_t}$ from the constant value $\alpha$  due to measurement errors and internal variability of the climate system. 
For the statistical model \eqref{eq:AForig}, it is straightforward to estimate the AF parameter $\alpha$ using the
sample mean of the data ${G_t}/{E_t}$, yielding the \emph{ratio-based estimator} as given by
\begin{align*}
\hat \alpha_1 = \frac{1}{T}\sum_{t=1}^T \frac{G_t}{E_t}.
\end{align*}

The model in equation \eqref{eq:AForig} expresses that, on average, the fraction $\alpha$ of emissions $E_t$ is absorbed in the atmosphere, resulting in atmospheric concentrations increasing with the amount $G_t$.
An alternative way to express this association between $G_t$ and $E_t$ is through the model formulation
\begin{align}\label{eq:AFnew}
G_t = \alpha E_t+ u^{(2)}_t,
\end{align}
for $t = 1, 2, \ldots, T$, where the disturbance $u^{(2)}_{t}$ is also a zero-mean error process. A model closely related to \eqref{eq:AFnew} has previously been used to reconstruct and predict CO$_2$ growth rates \citep[][]{JC2005,Betts2016}.  
The relationship between the disturbances in equations \eqref{eq:AForig} and \eqref{eq:AFnew} is given by $u^{(1)}_t = u^{(2)}_t\, / \, E_t$. 
Cointegration of $G_t$ and $E_t$ implies that $u_t^{(2)}$ is a stationary process.
Then, the parameter $\alpha$ can be estimated directly
using a simple least-squares calculation, yielding the \emph{regression-based estimator} as given by
\begin{align*}
\hat \alpha_2 = \left( \sum_{t=1}^T E_t^2 \right)^{-1} \sum_{t=1}^T E_t G_t.
 \end{align*}

\subsection*{Statistical properties of the two estimators}\label{sec:E}

In the case of independent data, it is known that the regression-based estimator $\hat \alpha_2$ is efficient, i.e., it has lower estimation uncertainty compared to, for example, the ratio-based estimator $\hat \alpha_1$ \citep[e.g.,][]{Cochran1977,DW1987}. We study this question for the case of cointegrated non-stationary time series in this paper, which is shown to be the relevant case for the AF in \ref{app:unitroot}.

In the Supplementary Information (Section \ref{app:aTheory}), we derive the asymptotic properties, that is, consistency and asymptotic normality, for the regression-based and the ratio-based estimators. These properties depend on the dynamics of CO$_2$ emissions, which are well-described by a random walk with drift over the sample 1959--2022, i.e. $E_t = E_0 + bt + x_t$, where $b>0$ and $x_t$ is a random walk (Supplementary Information \ref{app:emissions}).\footnote{The variability of the differences in CO$_2$ emissions increases in the early 1990s (Supplementary Information, Figure \ref{fig:dE}). This is most likely due to increased variability of emission estimates from land-use and land-cover change, starting in the early 1990s (Supplementary Information \ref{app:diffdata}). Below, we analyze the more recent subsample 1992--2022 as a robustness check, and we find that the results are similar to those obtained for the full   sample.} 
In Proposition \ref{prop:park} of Section \ref{app:aTheory}, we show that the regression-based estimator $\hat \alpha_2$ converges to the data-generating AF $\alpha$ at rate $T^{3/2}$, where $T$ is the sample size. The ratio-based estimator $\hat \alpha_1$, on the other hand, converges to the data-generating AF $\alpha$ at the slower rate $T$, as is shown in Proposition \ref{prop:a1}. This implies that, when the sample size $T$ is large, the estimation uncertainty in the regression-based estimator will be lower than in the ratio-based estimator, as in the case of independent data.

If the process $E_t$ has positive probability density at zero, then the ratio-based estimator $\hat \alpha_1$ does not have a finite mean or variance, as Proposition \ref{prop:moments} shows. This follows directly from the definition of $\hat \alpha_1$ as the sample mean of $G_t/E_t$: if values $E_t = 0$ have positive probability in the sample space, then the ratio $G_t/E_t$ is not integrable on that space.

The maintained model assumption $E_t = E_0 + bt + x_t$ for CO$_2$ emissions implies positive probability density at zero, if $x_t$ is a random walk of, for example, Gaussian increments. On the sample 1959--2022, however, the trend $E_0 + bt$ is much larger in magnitude than the random walk $x_t$, and it is not unreasonable to assume $x_t = 0$ for theoretical purposes. In this case, the ratio-based estimator $\hat \alpha_1$ has some standard statistical properties: It is an unbiased estimator of $\alpha$, that is, the mean of $\hat \alpha_1$ equals  $\alpha$, and the variance of  $\hat \alpha_1$ has a simple expression that can easily be estimated. However, we show in Proposition \ref{prop:noCLT}(i) that even in this case, a central limit theorem does \emph{not} hold in general. The ratio-based estimator has a limiting Gaussian distribution only if we additionally assume that  $u_t$ is Gaussian, shown in Proposition \ref{prop:noCLT}(ii). In contrast, the regression-based estimator $\hat \alpha_2$ follows a central limit theorem with a limiting Gaussian distribution and the derivation does not require this additional assumption, as shown in Proposition \ref{prop:park}.

Although the theoretical results show that the regression-based estimator $\hat \alpha_2$ is asymptotically, i.e., for sufficiently large $T$, more precise than the ratio-based estimator $\hat \alpha_1$, it is an empirical question which estimator is more precise in finite samples. In the next section, we estimate the variances of the regression-based and the ratio-based estimators on the historical sample and compare their magnitudes.

\section*{Estimating the airborne fraction over 1959--2022} \label{sec:estimate}

We use time series data on yearly changes in atmospheric CO$_2$ ($G_t$), yearly CO$_2$ emissions from fossil fuels ($E_t^{FF}$), and yearly CO$_2$ emissions from land-use and land cover change ($E_t^{LULCC}$), for the sample 1959--2022. Total anthropogenic CO$_2$ emissions are then $E_t = E_t^{FF} + E_t^{LULCC}$. The data series are measured in gigatonnes of carbon per year (GtC/yr), obtained from the Global Carbon Project\footnote{Data from the Global Carbon Project can be found here: \url{https://www.icos-cp.eu/science-and-impact/global-carbon-budget/2023}, last accessed June 17, 2024. The time series $E_t^{FF}$ includes the cement carbonation sink, as described in \cite{GCB2023}. } and presented in Figure~\ref{fig:data}a)--b).

The ratio-based estimate $\hat \alpha_1$ and the regression-based estimate $\hat \alpha_2$ are obtained from least-squares regressions applied to models \eqref{eq:AForig} and \eqref{eq:AFnew}, respectively. The fits are shown in Figure~\ref{fig:data}c) and \ref{fig:data}e) and the associated estimated residuals $\hat u_t$ are shown in \ref{fig:data}d) and \ref{fig:data}f), all as blue lines. 
To account for possible serial correlation and heteroskedasticity in the model errors $u_t$,
we calculate standard errors using the heteroskedasticity and autocorrelation consistent (HAC) estimator of
\cite{newey1987simple}. The results are displayed in the first two columns of Table \ref{tab:analysis}.
The estimates largely agree on the magnitude of the AF,  $\hat \alpha_1 = 43.86\%$ and $\hat \alpha_2 = 44.78\%$.
However,
the standard error of $\hat \alpha_2$ is $11\%$ lower than the standard error of $\hat \alpha_1$,
showing that the faster convergence rate of this estimator ($T^{3/2}$ versus $T$)
outweighs the fact that the error
process $u_t^{(2)}$ in \eqref{eq:AFnew} has a larger variance than the error process $u_t^{(1)}$
in \eqref{eq:AForig}. In particular, the  
estimated standard deviations (SDs) of these model errors are
$\widehat{SD}(u_t^{(1)}) = 0.13$ and
$\widehat{SD}(u_t^{(2)}) = 0.91$.
The discrepancy is due to the different nature of the two models where
$u_t^{(1)} = u_t^{(2)}/E_t$, with $E_t \gg 1$ in the sample 1959--2022.

\begin{table}[h]
\caption{\it Least-squares regression output of the linear models \eqref{eq:AForig}--\eqref{eq:AFnew2} for the Global Carbon Project data on the full sample 1959--2022 (left panel) and the subsample 1992--2022 (right panel). Estimates of the AF $\alpha$ are denoted by ``$\hat \alpha$''. The table reports standard errors ``SE($\hat \alpha$)'' and standard errors relative to SE($\hat \alpha_1$) from model (1). Relative SE less than one indicates lower estimation uncertainty than the ratio-based estimator. The $95\%$ confidence interval for the AF $\alpha$, based on the Gaussian distribution, is denoted by ``$CI_{95\%}(\alpha)$'', the estimated standard deviation of the error $u_t$ by ``$\widehat{SD}(u_t)$'', and the
coefficient of determination by ``$R^2$''.}
\vspace{4mm}
\centering
\scriptsize
\begin{tabular}{lcccc|cccc@{}}
& \multicolumn{4} {c} {\bfseries Full sample (1959-2022)} & \multicolumn{4} {c} {\bfseries Recent sample (1992-2022)} \\
\cmidrule{2-9}
 &  \eqref{eq:AForig} &  \eqref{eq:AFnew} &  \eqref{eq:AForig2} &  \eqref{eq:AFnew2}  &  \eqref{eq:AForig} &  \eqref{eq:AFnew} &  \eqref{eq:AForig2} &  \eqref{eq:AFnew2}   \\ \hline 
$\hat \alpha$    &           0.4386   & 0.4478 &   0.4716  &  0.4697                            &                       0.4456  &  0.4497 &   0.4626  &  0.4613 \\
SE($\hat \alpha$)   &      0.0159 &   0.0141&    0.0126   & 0.0105                       &               0.0190 &   0.0157&    0.0124  &  0.0104  \\
Relative SE   &      1.0000 &   0.8895   & 0.7904    &0.6630                                &            1.0000   & 0.8247 &   0.6509  &  0.5464  \\
 $CI_{95\%}(\alpha)$  &  $ [ 0.4074, $ &    $ [0.4201, $ &    $ [0.4470, $  &    $ [0.4490,  $      &           $ [0.4083, $    &  $ [0.4190, $    &  $ [0.4384, $    &  $ [0.4409,  $\\ 
 &  \hspace{0.25cm} $  0.4697 ]$ &   \hspace{0.25cm}  $ 0.4755 ]$ &    \hspace{0.25cm} $  0.4962]$  &  \hspace{0.25cm}   $  0.4903]$                        &  \hspace{0.25cm} $  0.4828 ]$    &  \hspace{0.25cm} $  0.4804]$    & \hspace{0.25cm}  $  0.4869  ]$    &  \hspace{0.25cm}  $   0.4816]$\\
$\widehat{SD}(u_t)$  &        0.1258  &  0.9088   & 0.0881   & 0.6292                       &                0.1065  &  0.9309  &  0.0662  &  0.5891 \\
$R^2$ &         0 &   0.5863  &  0.5258  &  0.8080                                      &                  0  &  0.3558   & 0.6391 &   0.7592 \\
\cmidrule{2-9}
 ENSO included & No & No & Yes & Yes & No & No & Yes & Yes \\
VAI included  & No & No & Yes & Yes& No & No & Yes & Yes  \\ \hline
\end{tabular}
\label{tab:analysis}
\end{table}

By introducing covariates in  models \eqref{eq:AForig} and \eqref{eq:AFnew}, we can
reduce the variance of the error processes $u_t$ and thus achieve more precise estimates of the AF $\alpha$. For example, it is common practice in the literature to control for the effects of El Ni\~no ($ENSO_t$) and volcanic activity
($VAI_t$), \citep[e.g.][]{Raupach2008,vMa2023}.\footnote{VAI data for volcanic activity are obtained from \cite{Ammann2003}. ENSO data are constructed from the Ni\~no 3 SST Index of
the National Oceanic and Atmospheric Administration (NOAA), available at \url{https://psl.noaa.gov/gcos_wgsp/Timeseries/Data/nino3.long.anom.data},
last accessed June 17, 2024. We have converted monthly ENSO data into a yearly time series of September-August ENSO means \citep[][]{Jones2001}.
This 4-month lag provides the best fit between $G_t$ and $ENSO_t$.
The slight trend in yearly ENSO data is removed
so that it has no
impact on the AF estimates. The data are shown in Figure \ref{fig:data_all} of the Supplementary Information.}  Hence, we consider the models
\begin{align}
\frac{G_t}{E_t} &= \ \alpha \ + \ \tilde\gamma_1 ENSO_t + \tilde\gamma_2 VAI_t + u^{(3)}_{t}, \label{eq:AForig2} \\
G_t  &= \alpha E_t  + \gamma_1 ENSO_t + \gamma_2 VAI_t + u^{(4)}_t, \label{eq:AFnew2}
\end{align}
for $t = 1, 2, \ldots, T$, where $\tilde\gamma_i$ and $\gamma _i$, for $i=1,2$,
are regression coefficients,
and the model errors $u^{(j)}_{t}$ follow a zero-mean error process, for $j=3,4$.
For both models, the coefficients can be estimated using least-squares regression.
Let $\hat \alpha _3$ and $\hat \alpha _4$ denote the least-squares estimators of $\alpha$
from models \eqref{eq:AForig2} and \eqref{eq:AFnew2}, respectively.
The estimation results relevant for the AF $\alpha$ 
are presented in columns labeled as (\ref{eq:AForig2}) and (\ref{eq:AFnew2}) of Table \ref{tab:analysis}, while
the estimation results for the ENSO and VAI coefficients $\gamma_i$ and $\tilde \gamma_i$, for $i = 1,2$,
are reported in Table \ref{tab:analysis_all} of the Supplementary Information.
The regression fits of models \eqref{eq:AForig2} and \eqref{eq:AFnew2} are presented as in Figure~\ref{fig:data}c) and \ref{fig:data}e)
and their associated estimated residuals $\hat u_t$ in Figure~\ref{fig:data}d) and \ref{fig:data}f), all as red dashed lines. 
The results show that controlling for the effects of ENSO and volcanic activity increases the
estimate of the AF considerably,
resulting in $\hat \alpha_3 = 47.16\%$ and $\hat \alpha_4 = 46.97\%$.
The estimates of the standard deviations of the error terms $(\widehat{SD}(u_t^{(3)}) =  0.09$
and $\widehat{SD}(u_t^{(4)})  = 0.63)$ decrease substantially compared to those of models \eqref{eq:AForig} and \eqref{eq:AFnew}, indicating that the covariates ENSO and VAI explain much variation in the data. This is corroborated by the coefficient of determination ($R^2$) values reported in Table  \ref{tab:analysis}.\footnote{The $R^2$ of model (\ref{eq:AForig}) equals zero by construction,
since this model only features an intercept.} 
The decreased variance of the residuals from models \eqref{eq:AForig2} and \eqref{eq:AFnew2}
imply that their estimates of the constant AF $\alpha$ are more precise.
The standard error of the estimate $\hat \alpha_4$ is approximately $16\%$ lower than the one of
estimate $\hat \alpha_3$ and approximately $34\%$ lower than the one of
the conventional estimate $\hat \alpha_1$.
Our preferred estimate $\hat \alpha_4$ results in an estimated AF of $47.0\%$ ($\pm 1.1\%$; $1 \sigma$) with an associated $95\%$ confidence interval of $[44.9\%,49.0\%]$.
The slightly increased AF estimates for models \eqref{eq:AForig2} and \eqref{eq:AFnew2}, compared to the models \eqref{eq:AForig} and \eqref{eq:AFnew} without covariates, confirm a similar finding in \cite{Betts2016}.

As a robustness check, the right panel of Table \ref{tab:analysis} presents
the results for the more recent subsample 1992--2022;
Table \ref{tab:analysis1992} in the Supplementary Information contains the corresponding coefficient estimates for ENSO and VAI. Our conclusions for the full sample are corroborated by the results for the recent sample. All estimates of the AF $\alpha$ from the subsample are within the respective confidence bands of the estimates from the full sample, while the reductions in uncertainty from including the covariates and from using the regression-based estimator are similar.

\section*{Estimating the airborne fraction over 2023--2100}

The approximate constancy of the AF over the historical period 1959--2022,
as documented in the literature and confirmed by the cointegration analysis in this study,
can be taken as the result of a near exponential growth in emissions and an
approximately linear response of the carbon sinks to atmospheric concentrations \citep[][]{Raupach2014}. In scenarios describing the future, for example when emissions are declining, the AF is expected to depart from constancy and may vary over time \citep[e.g.,][]{Jones2013,Pressburger2023}. This motivates the specification of a time-varying AF $\alpha=\alpha_t$, with $\alpha_t$ denoting the AF in year $t$, i.e.
the fraction of emissions ($E_t$) added to the atmosphere ($G_t$) in year $t$. 

The ratio-based model can then be written as
\begin{align*}
\frac{G_t}{E_t} = \alpha_t + u_t^{(1)},
\end{align*}
where $\alpha_t$ is a yearly time-varying coefficient. 
The ratio $G_t/E_t$ may be used to track the amount of emitted CO$_2$ that remains airborne, and hence,
the estimate $\hat \alpha_{1,t} = G_t / E_t$ can be regarded as an appropriate but very noisy time-varying AF.
A possible way to reduce the noise in AF is to apply a local smoothing operation, e.g. a two-sided moving average filter. In any case, filtered or not, the variability of the ratio $G_t/E_t$ will be amplified when future emissions $E_t$ start to approach zero. Another possible solution for noise reduction previously suggested in the literature, is to use the cumulative AF (CAF) in place of the yearly AF \citep[e.g.,][]{Jones2013,Jones2016,Liddicoat2021}. However,  due to its cumulative nature, the CAF can be slow to detect changes in the behavior of the carbon sinks, making it less useful for the purpose of analyzing a time-varying AF (Supplementary Information \ref{app:caf}).

When considering the regression-based model with a time-varying AF $\alpha_t$, we obtain
\begin{align}\label{eq:TVa}
G_t = \alpha_t E_t + u_t^{(2)}.
\end{align}
A versatile way of treating such a time-varying regression model is to assume random walk dynamics for $\alpha_t$ and estimate it by means of a recursive regression filter, such as the Kalman filter and smoother \citep[][\S 3.6.1, \S 4.3, \S 4.4]{Durbin2012}. This approach yields the minimum mean-squared error estimator $\hat \alpha_{2,t}$, and it does not suffer from the deficiencies of the time-varying ratio-based estimator $\hat \alpha_{1,t}$.
In the regression-based model \eqref{eq:TVa}, $\alpha_t$ is multiplied by $E_t$.
When emissions turn negative for the first time in some year $t$, i.e. $E_t<0$, we let 
$\alpha_t$ be reflected around one.
For this purpose, we adjust the random walk specification at year $t$ with a one-time instantaneous shift in $\alpha_t$ when $E_t<0$ for the first time.
Section \ref{app:KF} of the Supplementary Information provides further motivation and technical detail on this procedure.

To study the performance of the ratio-based $\hat \alpha_{1,t}$ and the regression-based $\hat \alpha_{2,t}$ estimators in situations where the AF is changing over time, we apply the two estimators to output from the MAGICC reduced-complexity climate model \citep[][]{MAGICC}. We let MAGICC produce future trajectories of $G_t$ and $E_t$ for $t = 2023, 2024, \ldots, 2100$ according to the Shared Socioeconomic Pathways \citep[SSPs;][]{Riahi2017}.\footnote{The SSP scenarios can be run in MAGICC in a web browser via the link \url{https://live.magicc.org/scenarios/bced417f-0f7f-4bb7-8359-792a0a8b0368/overview}. Last accessed June 17, 2024.}  We focus on the so-called SSP1-2.6 scenario, which is a high mitigation scenario consistent with a forcing level of $2.6$ Wm$^{-2}$ in the year $2100$ \citep{Riahi2017}. Results obtained from other SSP scenarios are similar to those reported below and are presented in Section \ref{app:SSP} of the Supplementary Information. Since MAGICC is a  deterministic model without a stochastic representation of the climate variables, the  trajectories of $G_t$ and $E_t$ generated by MAGICC are very smooth. To obtain output that resembles climate data, we perturb the trajectories of $G_t$ and $E_t$ by zero-mean Gaussian noise, where we set the variances equal to the estimates obtained on the historical data.
These simulated trajectories, together with the original output from MAGICC and the historical Global Carbon Project data 1959--2022, are shown in panels a) and b) of Figure \ref{fig:SSP126}. Its panel c) presents the historical ratio $G_t/E_t$ over 1959--2022 and the ratio-based and regression-based estimates $\hat\alpha_{t}$ of the time-varying airborne fraction over 2023--2100. 

The ratio-based estimate $\hat \alpha_{1,t} = G_t/E_t$ (blue) is a very noisy series as we have anticipated, especially when $E_t \approx 0$. In contrast, the regression-based estimate $\hat \alpha_{2,t}$ (red) evolves  over time in a stable fashion and shows sensible AF estimates, also when $E_t \approx 0$. A further benefit of the regression-based method is its delivery of confidence intervals for $\hat \alpha_{2,t}$ (shaded red area) which are not immediately available for the ratio-based estimator $\hat \alpha_{1,t}$. The covariates for El Ni\~no and volcanic activity can  readily be incorporated into the regression-based framework. 

In the SSP1-2.6 scenario studied here, the regression-based AF estimate $\hat \alpha_{2,t}$ remains roughly constant until 2050, after which it gradually declines towards zero.
In 2060, the atmospheric changes turn negative, resulting in a negative estimate of the AF, meaning that the sink uptake exceeds the emissions.
In 2077, the emissions turn negative as well, instigating the visible switch to a positive estimate of the AF.
The estimates of the AF exceed one from 2077 onwards, indicating that the sinks continue to absorb CO$_2$ even in this regime with
highly negative emissions. These findings can be contrasted with the analysis of the SSP1-1.9 scenario
(Supplementary Information, Figure  \ref{fig:SSP119}) which has a similar trajectory for the regression-based AF estimates
$\hat \alpha_{2,t}$ as for the SSP1-2.6 scenario, except that from 2080 onwards we have $\hat \alpha_{2,t}<1$ implying that
the sinks turn into carbon sources (releasing more carbon dioxide than they absorb).

\begin{figure}[t!]
    \centering
   \caption{\emph{Analysis of SSP1-2.6 data over the period $2023$--$2100$.
   a) Atmospheric concentration changes ($G_t$) for the historical period 1959--2022 (black) and the SSP period 2023--2100 (blue, magenta).
   b) Emissions ($E_t$) data. Magenta lines show the output from MAGICC; blue lines show the perturbed data.
   c) Ratio of atmospheric changes to emissions ($G_t/E_t$).
   The red line in c) is the estimated fraction of emissions ($E_t$) added to the atmosphere ($G_t$) in each year $t$,
   and, more specifically, it is the regression-based estimator $\hat \alpha_{2,t}$ of the time-varying airborne fraction $\alpha_t$,
   obtained from the Kalman smoother. Shaded area is a $95\%$ confidence band around $\hat \alpha_{2,t}$.} }
    \includegraphics[width=0.99\textwidth]{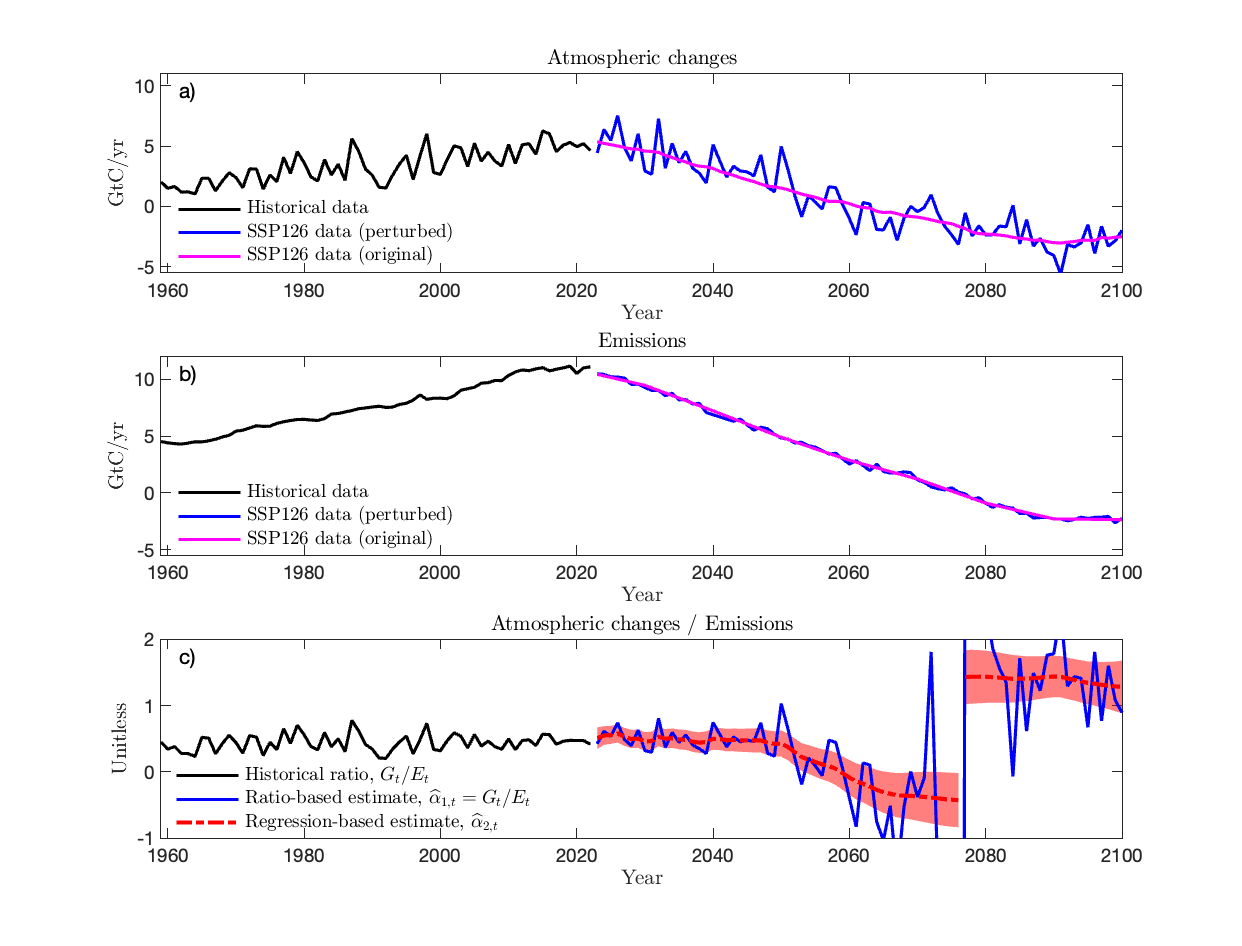}
    \label{fig:SSP126}
\end{figure}

\section*{Discussion}\label{sec:conclusion}

Our empirical findings present a slightly higher AF than the consensus estimate of $44\%$   \citep[][p. 676]{Canadell2021} and the cumulative airborne fraction
CAF$_t = 44.4\%$ (Supplementary Information \ref{app:caf})
obtained from the Global Carbon Project data. The regression-based estimate of the AF, using the 1959-2022 sample of the Global Carbon Project data and controlling for El Ni\~no and volcanic activity, is $47.0\%$ ($\pm 1.1\%$; $1 \sigma$), with a $95\%$ confidence interval of $[44.9\%,49.0\%]$.

When El Ni\~no and volcanic activity are excluded from the analysis, the estimate is $44.8\%$ ($\pm 1.4\%$; $1 \sigma$), which is more in line with the commonly reported results. When we apply the same analysis to two alternative data sets, we obtain slightly higher estimates of the AF than those from the Global Carbon Project (Supplementary Information \ref{app:diffdata}). The more recent 1992--2022 subsample yields an AF estimate of approximately $46\%$ ($\pm 1.0\%$; $1 \sigma$) for the Global Carbon Project data and slightly higher estimates for the two alternative data sets (Table \ref{tab:analysis} and Supplementary Information \ref{app:analysis1992}). To account for possible measurement error in $G_t$ and $E_t$,
we report Deming regressions \citep[][]{Deming1943} which are in line with the results reported so far (Supplementary Information \ref{app:deming}).
We may conclude that measurement error is not driving our results. 

To summarize the theoretical findings in our study, we conclude that the ratio-based estimator of the AF suffers from three main shortcomings. First, due to its definition as the ratio of changes in atmospheric concentrations to emissions, means and variances do not exist if zero emissions are possible. While this is of no concern on the historical sample, it is important when analyzing the AF on net-zero emissions scenarios. Studies of the past AF \citep[such as][among others]{Canadell2007,Raupach2008,LeQuere2009,Knorr2009,Ballantyne2012,Raupach2014,Ballantyne2015,keenan2016recent,BHK2019a,vMa2023} are most likely not influenced to any substantial degree by this issue, but studies of future low-emission scenarios are affected \citep[e.g.][]{Pressburger2023}. Second, we require stronger Gaussian assumptions on the distribution of the error process compared to the case of the regression-based estimator when a central limit theorem is invoked to compute confidence intervals and $p$-values based on a Gaussian distribution. Alternative methods such as the bootstrap can also be used for this purpose \citep[][]{Efron1993}. Again, studies on historical data are most likely not strongly affected by our findings, as Figure \ref{fig:data}d suggests Gaussian residuals.  Third, the ratio-based estimator converges to the data-generating AF at a slower rate than the regression-based estimator, even if zero-emissions are ruled out and errors are assumed to be normal. Both estimators converge faster than the common $\sqrt{T}$ rate due to the non-stationarity of the two yearly time series variables emissions ($E_t$) and changes in atmospheric concentrations ($G_t$) and their cointegration. The ratio-based estimator converges at rate $T$ and the regression-based estimator at rate $T^{3/2}$.

Our preferred regression-based estimator
in this study has standard statistical properties, such as existence of first and second moments, it is defined for zero emissions, and it converges to the data-generating AF at a fast rate. A central limit theorem applies without assuming Gaussianity of the regression error, and confidence levels and $p$-values can be computed in the usual way. Based on theoretical arguments, on a simulation study (Supplementary Information \ref{app:simstudy}), and on a historical sample of yearly data, we have shown that the regression-based estimator exhibits lower estimation uncertainty compared to the ratio-based estimator. Finally, we have argued that the regression-based estimator  can readily be generalized to a time-varying AF specification with its estimations by means of the Kalman filter and smoother.

Table \ref{tab:summaryTests} summarizes the statistical tests performed in this study. The main empirical findings are: (1) emissions and changes in atmospheric concentrations cointegrate on the historical sample 1959--2022 with a constant regression coefficient, motivating the model choice for our theoretical studies, (2) the regression errors appear Gaussian, (3) the regression-based estimator has lower standard errors on the historical sample and in simulations than the ratio-based estimator, and (4) the findings are qualitatively the same on a subsample of the last 31 years. Further, we have provided evidence that measurement error in the two time series is not driving the results.

The main advantage of the regression-based estimator for a historical sample analysis is increased precision. In cases of net-zero emission scenarios and of future data analysis  with emissions approaching zero, it will remain valid, unlike the ratio-based estimator.  To illustrate this advantage, we have simulated trajectories for emissions and changes in atmospheric concentrations over the period 2023--2100 consistent with SSP scenarios using the MAGICC reduced-complexity climate model.
We regard this development as a first step and consider the use of climate projections from the Coupled Model Intercomparison Project (CMIP) as next steps in our research agenda.

\begin{table}[ht]
\centering
\begin{tabularx}{\textwidth}{|X|X|X|}
\hline
\textbf{Test} & \textbf{Hypothesis of interest} & \textbf{Results} \\
\hline
Dickey-Fuller \citep[][]{dickey1979distribution}& Unit root in $G_t$, $E_t$ & Confirmed \\
\hline
Engle-Granger  \citep{EG1987}  & Cointegration of $G_t$, $E_t$ & Confirmed; AF constant on historical sample \\
\hline
Jarque-Bera \citep{JarqueBera1987} & Normality of $u_t^{(2)}$ in $G_t = \alpha E_t + u_t^{(2)}$. & Confirmed \\
\hline
Sample split (1959--2022 vs. 1992--2022)& Robustness of results & Confirmed; AF constant on historical sample \\
\hline
Deming  \citep[][]{Deming1943} & Robustness to measurement error in variables & Confirmed \\
\hline
\end{tabularx}
\caption{Summary of statistical tests}
\label{tab:summaryTests}
\end{table}

\section*{Acknowledgements}
We thank Morten \O. Nielsen for helpful discussions regarding convergence of stochastic processes. M.B. gratefully acknowledges funding from the Independent Research Fund Denmark under grant no. 0219-00001B.

\section*{Author contributions}
All authors contributed equally to the paper.

\section*{Data availability}
Data are publicly available and can be found at 

\noindent \url{https://github.com/mbennedsen/Regression-Approach-to-CO2-Airborne-Fraction}.

\section*{Code availability}
MATLAB code for replication of all results in the main paper and Supplementary Information can be found at \url{https://github.com/mbennedsen/Regression-Approach-to-CO2-Airborne-Fraction}.

%

 {\small 
\bibliographystyle{chicago}
\bibliography{bhk_references}
}

    \newpage \clearpage
    
    \setcounter{section}{0}
    \renewcommand{\thesection}{S\arabic{section}}
    \renewcommand{\thesubsection}{S\arabic{section}.\arabic{subsection}}
    \setcounter{equation}{0}
    \setcounter{figure}{0}
    \setcounter{table}{0}
    \setcounter{page}{1}
    \makeatletter
    \numberwithin{equation}{subsection}
    
    \renewcommand{\theequation}{S\arabic{section}.\arabic{equation}}
    \renewcommand{\thefigure}{S\arabic{figure}}
    \renewcommand{\thetable}{S\arabic{table}}
    \renewcommand{\theproposition}{S\arabic{proposition}}
    \renewcommand{\theremark}{S\arabic{remark}}
    
    \begin{center}
        \textbf{\large Supplementary Information \\ A Regression-based Approach to the CO$_2$ Airborne Fraction: \\ Enhancing Statistical Precision and Tackling Zero Emissions}
    \end{center}

%

\section{Theoretical analysis of the two estimators} \label{app:aTheory}
Consider the cointegrated system
\begin{align*}
y_t &= \alpha z_t + u_t, \\
z_t &= z_0 + b t + x_t, \\
x_t &=  \sum_{\tau=1}^t \xi_t, 
\end{align*}
where $z_0, b \in \mathbb{R}$, $u_t$ and $\xi_t$ are covariance-stationary error terms, and $t = 1,2, \ldots, T$. In the context of the analysis of the AF, we have $y_t = G_t$ and $z_t = E_t$. Recall the two estimators of $\alpha$:
\begin{align*}
\hat \alpha_1 &= \frac{1}{T} \sum_{t=1}^T \frac{y_t}{z_t}, \\
\hat \alpha_2 &= \left( \sum_{t=1}^T z_t^2 \right)^{-1} \sum_{t=1}^T z_t y_t.
\end{align*}
The second estimator $\hat \alpha_2$ has been thoroughly studied in the literature on cointegration. If $b = 0$, i.e. if there is no linear trend, then the stochastic trend $x_t$ in $z_t$ implies that $T (\hat \alpha_2 - \alpha)$ has a non-Gaussian limiting distribution \citep{PP1998}. If $b \neq 0$, i.e., if there is a linear trend present in $z_t$, then the linear trend dominates the stochastic trend, which implies that $T^{3/2}(\hat \alpha_2 - \alpha)$ has a Gaussian limit. These results are well-known, see e.g. \cite{Park1992}, but for convenience we state the result relevant for the specific setting in this paper.  

\begin{proposition}[\cite{Park1992}]\label{prop:park}
Let $b \neq 0$ and suppose that $\nu_t = (u_t,\xi_t)'$ is stationary and ergodic with zero-mean. Define $B_n(r) := n^{-1/2} \sum_{i=1}^{\lfloor n r \rfloor} \nu_i$ for $r \in [0,1]$. Suppose that $B_n  \stackrel{d}{\rightarrow} B$  such that $B$ is a Brownian motion with $Var(B_1)$ positive definite. Then
\begin{align*}
T^{3/2}(\hat \alpha_2 - \alpha) \stackrel{d}{\rightarrow} N(0,\sigma_2^2),
\end{align*}
as $T\rightarrow \infty$, where $\sigma^2_2>0$.
\end{proposition}

\begin{proof}
The assumptions made on $\nu_t$ ensure that Assumption 2.1 in \cite{Park1992} is satisfied. Since $b \neq 0$, Assumption 2.2 in \cite{Park1992} also holds. The result now follows from Lemma 3.1 and Corollary 3.2 in \cite{Park1992}.
\end{proof}

\begin{remark}
When $u_t$ is independent and identically distributed (iid), usual least-squares inference is valid, meaning that $Var(\hat \alpha_2)$ can be estimated by the usual least-squares estimator
\begin{align*}
\widehat{Var}(\hat \alpha_2) =  \left( \sum_{t=1}^T z_t^2 \right)^{-1}  T^{-1} \sum_{t=1}^T \left( y_t - \hat \alpha_2 z_t \right)^2.
\end{align*}
If $u_t$ may contain autocorrelation and heteroskedasticity, we can instead estimate $Var(\hat \alpha_2)$ using a heteroskedasticity and autocorrelation consistent (HAC) estimator, which is what we do in this study \citep{newey1987simple}.
\end{remark}

\begin{remark}
The conditions imposed on $\nu_t$ in Proposition \ref{prop:park} are rather weak and are satisfied for a large class of models. For instance, they are clearly satisfied if $u_t$ and $\xi_t$ are independent stationary and ergodic sequences. We refer to \cite{Phillips1986a} and references therein for various specific conditions ensuring the condition on $\nu_t$ holds.
\end{remark}

The properties of the estimator $\hat \alpha_1$ are less well known. Indeed, to the best of our knowledge, the statistical properties of $\hat \alpha_1$ have not been studied previously. To this end, we write
\begin{align*}
\hat \alpha_1 = \frac{1}{T} \sum_{t=1}^T \frac{y_t}{z_t} = \alpha +  \frac{1}{T} \sum_{t=1}^T \frac{u_t}{z_0 + b t + x_t} =  \alpha + T^{-1} \sum_{t=1}^T U_t, 
\end{align*}
where $U_t = u_t/(z_0 + b t + x_t)$. We first show that $\hat \alpha_1$ is a superconsistent estimator of $\alpha$.

\begin{proposition}\label{prop:a1}
Let $b \neq 0$ and assume that $u_t$ is iid with $\E(u_1) = 0$ and $Var(u_1)<\infty$, and that $x_t$ obeys a strong law of large numbers, i.e. $t^{-1}x_t \stackrel{a.s}{\rightarrow} 0$ as $t \rightarrow \infty$. Then, for all $\gamma > 0$, it holds that
\begin{align*}
T^{1-\gamma} (\hat \alpha_1 - \alpha) \stackrel{p}{\rightarrow} 0,
\end{align*}
as $T \rightarrow \infty$. 
\end{proposition}

\begin{proof}
Let $\epsilon, \delta>0$ be given. Suppose, for ease of notation, that $b>0$. The proof when $b<0$ follows similar arguments.  We want to show that
\begin{align}\label{eq:obj}
P(T^{1-\gamma} |\hat \alpha_1 - \alpha| \geq \epsilon) \leq \delta,
\end{align}
for all $T$ sufficiently large.

By assumption $x_t$ obeys a strong law of large numbers. By the continuous mapping theorem, the same holds for $|x_t|$, meaning that $P(A) = 1$, where
\begin{align*}
A = \{ \omega \in \Omega :\lim_{t\rightarrow \infty}  t^{-1} |x_t(\omega)| = 0\}.
\end{align*}
Let $T \geq 1$ and define the random variable 
\begin{align*}
\tilde x_T := \sup_{t \geq T} t^{-1} |x_t|.
\end{align*}
For  all $\omega \in A$, we have $\lim_{t\rightarrow \infty} t^{-1} |x_t (\omega) |= 0$ and thus $\limsup_{t \rightarrow \infty}  t^{-1} |x_t (\omega)| = \lim_{t\rightarrow \infty} t^{-1} |x_t (\omega)| = 0$. By the definition of limes superior, we get
\begin{align*}
\lim_{T\rightarrow \infty} \tilde x_T(\omega) := \lim_{T\rightarrow \infty}  \sup_{t \geq T} t^{-1} |x_t (\omega)|  = \limsup_{t \rightarrow \infty}  t^{-1} |x_t (\omega)| = \lim_{t\rightarrow \infty} t^{-1} |x_t (\omega)| =0.
\end{align*}
Since $P(A) = 1$, we conclude that $\tilde x_T \stackrel{a.s.}{\rightarrow} 0$ as $T \rightarrow \infty$, which implies that 
\begin{align}\label{eq:xtilde p}
\sup_{t \geq T} t^{-1} |x_t| = \tilde x_T \stackrel{p}{\rightarrow} 0
\end{align}
as $T \rightarrow \infty$. Let $\nu \in (0,1)$ and $T \geq 1$, and define the set
 \begin{align*}
A_T = \left\{  \sup_{t \geq T} t^{-1} |x_t |  \leq b \nu/2 \right\}.
\end{align*}
The convergence in \eqref{eq:xtilde p} implies that there exists a $T_1 \geq 1$ such that $P(A_T) \geq 1- \delta/2$  for all $T \geq T_1$. Likewise, since $z_0 \in \mathbb{R}$ is a constant, there exists a $T_2 \geq 1$ such that $\frac{|z_0|}{T} \leq b \nu/2$. Set $T_3 = \max (T_1,T_2)$, and write
\begin{align*}
T^{1-\gamma} (\hat \alpha_1 - \alpha) &=  T^{-\gamma} \sum_{t=1}^T \frac{u_t}{z_0 + b t + x_t}  \\
    & = T^{-\gamma} \sum_{t=1}^{T_3} \frac{u_t}{z_0 + b t + x_t} +  T^{-\gamma} \sum_{t=T_3+1}^T \frac{u_t}{z_0 + b t + x_t} \\
    &= T^{-\gamma} B_{1,T_3} + T^{-\gamma} B_{T_3+1,T},
\end{align*}
where $B_{s,r}:= \sum_{t=s}^{r} \frac{u_t}{z_0 + b t + x_t}$. Since $T_3$ is fixed, there are only finitely many terms in $B_{1,T_3}$, and hence $T^{-\gamma} B_{1,T_3}  \stackrel{p}{\rightarrow} 0$ as $T \rightarrow \infty$.  For $B_{T_3+1,T}$, we write
\begin{align*}
 T^{-\gamma} |B_{T_3+1,T}| &= T^{-\gamma} \left| \sum_{t=T_3+1}^T \frac{u_t}{z_0 + b t + x_t} \right| \\
    &= T^{-\gamma} \left|\sum_{t=T_3+1}^T \frac{u_t}{bt} \frac{1}{1 + z_0 b^{-1}t^{-1}  + x_t b^{-1} t^{-1}} \right| \\
        & \leq T^{-\gamma} \sum_{t=T_3+1}^T \frac{|u_t|}{b t} \frac{1}{|1 + z_0 b^{-1}t^{-1}  + x_t b^{-1} t^{-1}|}.
\end{align*}
 Assuming that we are on $A_{T_3}$, we therefore have       
\begin{align*}
         T^{-\gamma} |B_{T_3+1,T}|  \leq   T^{-\gamma} \sum_{t=T_3+1}^T \frac{|u_t|}{b t} \frac{1}{|1 - \nu|}.
\end{align*}
This last term converges to zero in $L^2$. To see this, note that
\begin{align*}
Var\left(    T^{-\gamma/2} \sum_{t=T_3+1}^T \frac{|u_t|}{b t} \frac{1}{|1 - \nu|}  \right) &= T^{-\gamma} \sum_{t=T_3+1}^T \frac{Var(u_t)}{b^2 t^2} \frac{1}{(1 - \nu)^2} \\
&\sim T^{-\gamma} Var(u_1) \sum_{t=T_3+1}^\infty \frac{1}{b^2 t^2} \frac{1}{(1 - \nu)^2}, 
\end{align*}
where the sum converges. This shows that $ T^{-\gamma/2} \sum_{t=T_3+1}^T \frac{|u_t|}{b t} \frac{1}{|1 - \nu|}$ converges in $L^2$ to a constant and thus that $T^{-\gamma} |B_{T_3+1,T}|  \stackrel{L^2}{\rightarrow} 0$  as $T \rightarrow \infty$. We conclude that $ T^{-\gamma} B_{T_3+1,T}  \stackrel{p}{\rightarrow} 0$ as $T \rightarrow \infty$. Since this holds on $A_{T_3}$ and $P(A_{T_3}) \geq 1- \delta/2$, then, by choosing $T \geq T_3$ large enough, \eqref{eq:obj} holds, which completes the proof.

\end{proof}

\begin{remark}
It is straightforward to  relax the assumption that $u_t$ is iid, made in Proposition \ref{prop:a1}. In practice, we use a HAC estimator to estimate the variance of our estimators of $\hat \alpha_2$ to control for possible autocorrelation and heteroskedasticity in $u_t$.
\end{remark}

\begin{remark}
The assumption in Proposition \ref{prop:a1} that $x_t = \sum_{i=1}^t \xi_i$ obeys a strong law of large numbers is satisfied under quite mild conditions. For instance, it is satisfied if $\xi_t$ is iid with finite variance (Kolmogorov's Law of Large Numbers), $\xi_t$ is stationary and ergodic \citep[the Ergodic Theorem, e.g.][Theorem 9.5.5, p. 487]{KarlinTaylor1975}, $\xi_t$ is a mixingale \citep[e.g.][]{deJong1995,Davidson1997},  or  under certain mixing conditions on $\xi_t$ \citep[e.g.][]{Fazekas2001,Kuczmaszewska2005}. 
\end{remark}

Although $\hat \alpha_1$ is a superconsistent estimator of $\alpha$, under what are arguably mild conditions, it is the case that $U_t$, and consequently $\hat \alpha_1$, does not possess any moments. 
This is, for example, the case if the denominator $z_0 + b t + x_t$ is a continuous random variable with positive probability density at zero. In such circumstances, $\E(\hat \alpha_1)$ and $Var(\hat \alpha_1)$ are not well-defined and talking about whether the estimator is, e.g., unbiased is meaningless. 
Although the results of non-existence of moments is well-known  \citep[e.g.][]{PiegorschCasella1985}, it may not be immediately obvious. We therefore state it as a proposition and provide a straightforward proof.
\begin{proposition}\label{prop:moments}
Let $X$ be a random variable with continuous probability density function $f(x)$ such that $f(0)>0$. Then $\E(X^{-\beta})$ does not exist for any $\beta\geq 1$.
\end{proposition}
\begin{proof}
Note that because $f$ is continuous and has $f(0)>0$, then there exists an $\epsilon>0$ such that $f(x)>0$ for all $x \in [0,\epsilon]$. (If $X$ cannot take positive values, then a similar arguments holds, but for $x \in [-\epsilon,0]$.) Let $A := \inf_{x \in [0,\epsilon]} f(x) > 0$ and write
\begin{align*}
\E(|X|^{-\beta}) = \int_{-\infty}^{\infty} |x|^{-\beta} f(x) dx \geq  \int_{0}^{\epsilon} |x|^{-\beta} f(x) dx \geq A \int_{0}^{\epsilon} x^{-\beta} dx = \infty
\end{align*}
when $\beta \geq 1$.
\end{proof}

\begin{remark}
Statistical issues with studying a ratio with a stochastic variable in the denominator have been encountered before in climate science. \cite{RoeBaker2007} discuss the ``long tails'' of the probability density function (PDF) for the climate sensitivity, i.e. the temperature response to a doubling of atmospheric CO$_2$ over its preindustrial value. They formulate the problem as $\Delta T = \lambda \cdot \Delta R_f$, where $\Delta T$ is the climate sensitivity and $\Delta R_f$ is the change in downward radiative flux, resulting from a doubling of CO$_2$ concentrations. The unknown parameter $\lambda$ is given by $\lambda_0 /(1-f)$, where $\lambda_0$ is a ``reference climate sensitivity'' and $f$ is the ``total feedback factor''. It is clear that if $f$ is a random variable such that the PDF of $f$ is continuous with positive density at $1$, then we are in a situation akin to the one studied in this paper and in Proposition \ref{prop:moments}. In fact, \cite{RoeBaker2007} rely on a number of published studies on the climate sensitivity and feedbacks in the climate system and argue that such a situation is likely. In their calculations, they assume that the distribution of $f$ is Gaussian, which satisfies the conditions of a continuous PDF with positive density at $1$. This leads to very long tails in the PDF of the climate sensitivity $\Delta T$ and, as our Proposition \ref{prop:moments} shows, non-existence of moments.
\end{remark}

In the setting of this paper, we find that the deterministic term $z_0 + b t$ is quite large compared to the stochastic term $x_t$, implying that any probability density of the denominator $z_0 + bt + x_t$ at zero is likely very small. Hence, we proceed to study the estimator $\hat \alpha_1$ under the assumption that $x_t = 0$, i.e. that $z_t = z_0 + bt$ is a deterministic linear time trend. This avoids the problem of dividing by a random variable and thus of having a positive probability density at zero for the term in the denominator. Assume, for simplicity, that $u_t$ is an iid sequence. In this case, we find that
\begin{align*}
\E(\hat \alpha_1) = \alpha,
\end{align*}
and
\begin{align*}
Var(\hat \alpha_1) = T^{-2} \sum_{t=1}^T \E(U_t^2) = T^{-2}  \sigma_u^2 \sum_{t=1}^T \frac{1}{(z_0 + bt)^2} \sim T^{-2} \sigma_u^2 c,
\end{align*}
as $T \rightarrow \infty$, where $c =  \sum_{t=1}^\infty \frac{1}{(z_0 + bt)^2}$ is a finite constant. Hence, in this case, we see that $\hat \alpha_1$ is an unbiased estimator of $\alpha$, converging at the rate $T^{-1}$. However, the estimator does not obey a central limit theorem. To see this, write
\begin{align}\label{eq:CLTa1}
T (\hat \alpha_1 - \alpha ) = \sum_{t=1}^T \frac{u_t}{z_0 + b t},
\end{align}
and observe that $\frac{u_t}{z_0 + b t} \stackrel{p}{\rightarrow} 0$ as $t \rightarrow \infty$, meaning that the first terms in the sum in \eqref{eq:CLTa1} dominate the entire sum, thus making central limit theorem-type results invalid. That is, we cannot conclude that $\hat \alpha_1$ is (asymptotically) Gaussian unless we impose an additional assumption that the error term $u_t$ is itself Gaussian.  For convenience, we again state these results formally.

\begin{proposition}\label{prop:noCLT}
Let $b \neq 0$ and assume that $u_t$ is iid with $\E(u_1) = 0$ and $Var(u_1)<\infty$ and consider the sequence of random variables for $C_T = \sum_{t=1}^T \frac{u_t}{z_0 + bt}$, $T= 1, 2, \ldots$. The following holds:

\begin{enumerate}
\item[(i)] A central limit theorem does not hold in general for $C_T$, i.e. we cannot conclude that $C_T$ has a Gaussian limit as $T \rightarrow \infty$.
\item[(ii)] If $u_t \sim N(0,\sigma_u^2)$ with $\sigma_u^2>0$, then
\begin{align*}
C_T \sim N(0,\sigma_{1,T}^2),
\end{align*}
where $\sigma_{1,T}^2 = \sigma_u^2\sum_{t=1}^T (z_0 + bt)^{-2}$, and, as $T \rightarrow \infty$,
\begin{align*}
C_T  \stackrel{d}{\rightarrow} N(0,\sigma_{1}^2),
\end{align*}
where $\sigma_{1}^2 = \sigma_u^2\sum_{t=1}^\infty (z_0 + bt)^{-2}$.
\end{enumerate}

\end{proposition}
\begin{proof}
(i) To see that a central limit theorem cannot hold in general, suppose that $u_t$ is such that $\E(u_t^3) = s \neq 0$. Then
\begin{align*}
\E(C_T^3) =   \sum_{t=1}^T \frac{\E(u_t^3)}{(z_0 + b t)^3} = s \sum_{t=1}^T \frac{1}{(z_0 + b t)^3} \rightarrow s  c_3 \neq 0,
\end{align*}
as $T \rightarrow \infty$, where $c_3 =  \sum_{t=1}^\infty \frac{1}{(z_0 + b t)^3}< \infty$ and $c_3 \neq 0$. This shows that the asymptotic skewness of $C_T$ is not zero, implying that the limiting distribution of $C_T$ is not Gaussian.

(ii) That $C_T \sim N(0,\sigma_{1,T}^2)$ follows trivially by the assumption that $u_t$ is iid Gaussian. The same holds for the second result. Formally, we may deduce it using characteristic functions. For $\theta \in \mathbb{R}$, we have
\begin{align*}
\E\left( e^{i\theta C_T}\right) = \prod_{t=1}^T \E\left(e^{i  \frac{\theta}{z_0 + bt} u_t}\right) = \prod_{t=1}^T e^{-\frac{1}{2}  \frac{\theta^2 \sigma_u^2 }{(z_0 + bt)^2} } =  e^{-\frac{1}{2}  \theta^2 \sigma_u^2 \sum_{t=1}^T \frac{1}{(z_0 + bt)^2} }  \rightarrow  e^{-\frac{1}{2} \theta^2 \sigma_1^2},  
\end{align*}
as $T \rightarrow \infty$, which concludes the proof.
\end{proof}

In the case of the conventional estimator  $\hat \alpha_1$, Proposition \ref{prop:noCLT} shows that if we want to base inference  on the normal distribution, we have to make distributional assumptions on the error term $u_t$. Indeed, if we are willing to assume that $u_t \sim N(0,\sigma_u^2)$, then we get
\begin{align*}
T (\hat \alpha_1 - \alpha ) \sim N(0,\sigma_{1,T}^2),
\end{align*}
and, as $T \rightarrow \infty$,
\begin{align*}
T (\hat \alpha_1 - \alpha ) \stackrel{d}{\rightarrow} N(0,\sigma_{1}^2),
\end{align*}
where $\sigma_{1,T}$ and $\sigma_{1}^2$ are as in Proposition  \ref{prop:noCLT}(ii). Note that if $u_t$ is iid, then $\sigma_{1,T}^2$ and $\sigma_{1}^2$ can be estimated directly from the sample variance of $y_t/z_t$. If $u_t$ may contain autocorrelation and heteroskedasticity, we can instead estimate $\sigma_{1,T}^2$ and $\sigma_{1}^2$ using a HAC estimator, which is what we do in this study.

\section{Further empirical results}\label{app:results}

\subsection{Unit root and cointegration tests}\label{app:unitroot}
Let $y_t$, $t = 1, 2, \ldots, T$, be a time series and $L \geq 0$ an integer. We consider the following model for $y_t$ for the null hypothesis,
\begin{align}\label{eq:adf null}
y_t = a + y_{t-1} +  \sum_{l=1}^L \beta_l \Delta y_{t-l} + \epsilon_t,
\end{align}
where $a,\beta_1,\ldots, \beta_L \in \mathbb{R}$ and $\epsilon_t$ is an error term. We use the convention that $\sum_{l=1}^L = 0$ if $L = 0$. The model for $y_t$ for the alternative hypothesis is
\begin{align}\label{eq:adf alternative}
y_t = b + c\cdot  t + \phi y_{t-1} +  \sum_{l=1}^L \beta_l \Delta y_{t-l} + \epsilon_t,
\end{align}
where $b,c \in \mathbb{R}$ and $\phi \in (-1,1)$.

We use the Augmented Dickey-Fuller test \citep[ADF;][]{dickey1979distribution}  to test the null ($H_0$) that a data series is generated by the model in \eqref{eq:adf null} against the alternative ($H_1$) that the data series is generated by \eqref{eq:adf alternative}. Depending on restrictions made on the three parameters $a,b,c$, this leads to three different tests:
\begin{enumerate}
\item
    AR: $a=b=c=0$. This is the test for $y_t$ having a unit root ($H_0$) against the alternative that $y_t$ is an AR process.
\item
    ARD: $a=c=0$. This is the test for $y_t$ having a unit root ($H_0$) against the alternative that $y_t$ is an AR process around a linear trend (i.e. that $y_t$ is ``trend-stationary'').
\item
    TS: This is the test for $y_t$ having a unit root and a linear trend ($H_0$) against the alternative that $y_t$ is a trend-stationary AR process.
    
\end{enumerate}
We implement the ADF test using the built-in function \texttt{adftest} in the MATLAB Econometrics Toolbox. The $p$-values for testing $H_0$ against $H_1$ for these three tests and the two data series $y_t = G_t$ and $y_t = E_t$ are shown in the top six rows of Table \ref{tab:unitroot}. For $y_t = E_t$ the $p$-values are very large regardless of the particular test conducted (AR, ARD, or TS) and the number of lags $L$ used. Hence, there is strong evidence that $E_t$ contains a stochastic trend. For $y_t = G_t$ the $p$-values are very large for the AR-test, regardless of the number of lags $L$ used. For the ARD-test, the $p$-values are below $10\%$ for $L = 0$ and $L = 1$ and very large for $L \geq 2$. For the TS-test, the $p$-values are below $10\%$ regardless of the number of lags $L$ used. Thus, the AR-test and ARD-test provide evidence that the time series of atmospheric CO$_2$ changes $G_t$ contains a stochastic trend, while the TS-test indicates that the trend might be modelled using a deterministic trend $b + c \cdot  t$. In summary, the ADF tests provide very strong evidence for the presence of trends in $E_t$ and $G_t$ and suggests that both series are non-stationary. 

We next employ the Engle-Granger test \citep{EG1987} to test whether $G_t$ and $E_t$ are cointegrated. Formally, we consider the model
\begin{align}\label{eq:eg null}
G_t = \alpha E_t + u_t
\end{align}
and test whether the error term $u_t$ is stationary. The parameter $\alpha$ is estimated by least-squares regression and the estimated residuals are computed,
\begin{align*}
\hat u_t = G_t - \hat \alpha E_t,
\end{align*}
where $\hat \alpha$ is the least-squares estimate of $\alpha$ obtained from \eqref{eq:eg null}. In a second step, we investigate whether the error terms $u_t$ are stationary by performing the AR-version of the ADF test on the estimated residuals $\hat u_t$. The Engle-Granger test adjusts the critical values of the ADF test to account for the fact that the error term $u_t$ is unobserved and that the  estimated residuals $\hat u_t$ are used in their place as data in the ADF test. We implement the Engle-Granger test using the built-in function \texttt{egcitest} in the MATLAB Econometrics Toolbox. The $p$-values for the Engle-Granger test are shown in the bottom row in Table \ref{tab:unitroot}. The null hypothesis of a unit root in $u_t$ is firmly rejected for all lags $L$ and we can therefore conclude that the two time series are cointegrated.

\subsection{Dynamics of CO$_2$ emissions over the 1959-2022 sample}\label{app:emissions}
The theoretical properties of the estimators of the AF depend on the dynamics of the emissions $E_t$, as discussed in Section \ref{app:aTheory}. Section \ref{app:unitroot} found that emissions are non-stationary. This section further explores the dynamics of the emissions over the sample period 1959--2022.

Figure \ref{fig:dE} motivates a random walk with drift as an appropriate model for emissions. Panel a) shows the differences of yearly CO$_2$ emissions, $\Delta E_t = E_t - E_{t-1}$, for $t=2,\ldots,T$. They are well described by a constant plus a stationary, zero-mean error term, say, $\Delta E_t = b + \xi_t$. Panel b) shows that the autocorrelation function of the differences in emissions is not statistically significant at all lags (except 0), and thus the error term $\xi_t$ is white noise.  Therefore, our model for CO$_2$ emissions $E_t$ is a random walk with drift $b$. We write $E_t = E_{t-1} + b + \xi_t$, or, equivalently, $E_t = E_0 + b\cdot t + x_t$, where $x_t = \sum_{\tau=1}^t \xi_\tau$ is the random walk process $x_t=x_{t-1}+\xi_t$ without drift, and $E_0$ are initial emissions.

The constant $b$ can be estimated directly as the sample mean of $\Delta E_t$, and on the sample 1959--2022, we find $\hat b = 0.1043$ (standard error $0.0241$). This estimate is close to the corresponding estimate found in \cite{BHK2023b} using a dynamic statistical system model for the Global Carbon Budget. A Jarque-Bera test \citep{JarqueBera1987} for normality of the estimated residuals $\hat \xi_t = \Delta E_t - \hat b$ results in a $p$-value below $0.1\%$, meaning that we can reject that $  \xi_t$ is a Gaussian process. 

Figure \ref{fig:dE}a) shows that the variability of the differences in CO$_2$ emissions increases in the early 1990s. This is most likely due to increased variability of emission estimates from land-use and land-cover change, starting in the early 1990s \citep[see Supplementary Information \ref{app:diffdata} and also][]{vMa2023}. In the main paper, we analyze the more recent subsample 1992--2022 as a robustness check, and we find results very similar to those obtained for the full 1959-2022 sample.

\begin{figure}[t!]
    \centering
    \caption{\emph{ Yearly changes in emissions ($\Delta E_t = E_{t} - E_{t-1}$) over the period $1960$--$2022$.}}
    \includegraphics[width=0.99\textwidth]{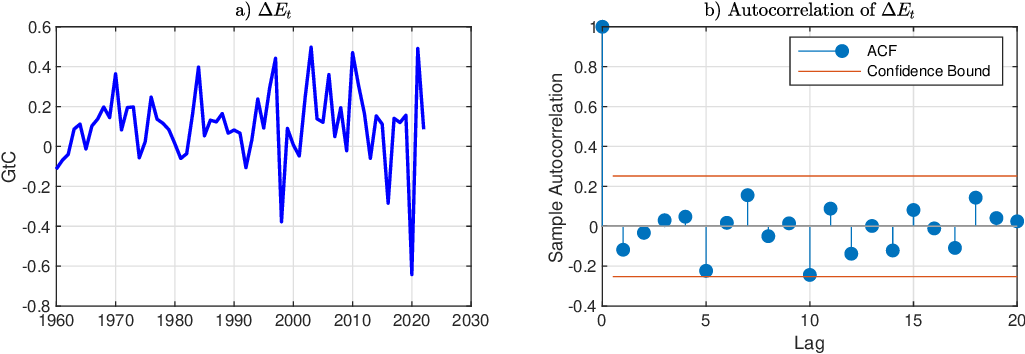}
    \label{fig:dE}
\end{figure}

\subsection{Analysis of alternative data sets}\label{app:diffdata}
We use data on yearly changes in atmospheric CO$_2$ ($G_t$), yearly CO$_2$ emissions from fossil fuels ($E_t^{FF}$), and yearly CO$_2$ emissions from land-use and land cover change ($E_t^{LULCC}$). Total anthropogenic CO$_2$ emissions are then $E_t = E_t^{FF} + E_t^{LULCC}$. The uncertainty in measurements of the AF stems in large part from uncertainties in the magnitude of LULCC emissions \citep{vMa2023}. We therefore follow  \cite{vMa2023} and estimate the AF using three different data sets for LULCC emissions (GCP, H\&C, vMa), while data for atmospheric CO$_2$ changes and CO$_2$ emissions from fossil fuels are obtained from the most recent edition of the Global Carbon Budget data set \citep{GCB2023}.\footnote{Data from the Global Carbon Project can be found at \url{https://www.icos-cp.eu/science-and-impact/global-carbon-budget/2023}, last accessed June 17, 2024. Data for the various LULCC time series can be found at \url{https://doi.org/10.7910/DVN/U7GHRH} and \url{https://github.com/mbennedsen/Regression-Approach-to-CO2-Airborne-Fraction}, last accessed June 17, 2024.} The GCP LULCC data are from the Global Carbon Project \cite{GCB2023}, the H\&C LULCC data are from \cite{HC23}, and the vMa LULCC data are from \cite{vMa2023}. We thus have three different data sets, only differing in the LULCC time series. The data are available for the period $1959$--$2022$ (Figure~\ref{fig:data_all}).

Estimation results for $\alpha$ using the alternative data sets are reported in Table \ref{tab:analysis_all}. For ease of comparison, we also include the results presented in Table \ref{tab:analysis} using the GCP data set. The table shows that the overall estimates of $\alpha$ vary across the three data sets, with the GCP data generally resulting in  lower estimates of the AF than the H\&C and vMa data, a result of the H\&C and vMa LULCC data being significantly below the GCP LULCC data over 1959--2022 (Figure \ref{fig:data_all}). The reduction in uncertainty obtained by using $\hat \alpha_2$ instead of $\hat \alpha_1$ is of the same magnitude across the three data sets, however. 

%

\subsection{Analysis of the subsample 1992--2022}\label{app:analysis1992}
In this section, we analyze the AF over the recent subsample 1992-2022 for all three data sets. The subsample analysis is motivated by the contradictory findings of trends in the AF mentioned in the introduction, the possibility of a structural break in the data around $1990$ \citep[][]{BHK2023a}, substantial disagreements in land-use and land-cover change emissions data before $1992$ \citep[][see also Figure \ref{fig:data_all}]{vMa2023}, and the presumption of better data quality in the latter part of the sample \citep[e.g.][]{Hill1999}. 

The estimation results are shown in Table  \ref{tab:analysis1992}. 
The results show estimates of $\alpha$ that are within the confidence intervals from the full-sample analysis (Table \ref{tab:analysis_all}), although the estimates coming from the most recent sample 1992--2022 are slightly below those from the entire sample 1959--2022.

We find that the regression-based estimator $\hat \alpha_2$ improves the precision, measured by the standard error of the estimator, by approximately $18\%$ compared to the conventionally used ratio-based estimator $\hat \alpha_1$. Including data on ENSO and volcanic activity in the analysis, the regression-based estimator $\hat \alpha_4$ improves the precision with approximately $45\%$, compared to the ratio-based  estimator of the AF $\hat \alpha_1$. Using data from the Global Carbon Project \citep[e.g.][]{GCB2023}, our best estimate of the AF is $46.1\%$ with an associated standard error of $1.0\%$, resulting in a $95\%$ confidence interval of $[44.1\%, 48.2\%]$ for the AF.  As was the case for the full sample studied in Section \ref{app:diffdata}, the alternative LULCC data results in higher estimates of the AF than the LULCC coming from the GCP data. 

\subsection{Accounting for potential measurement errors in the data}\label{app:deming}
It is likely that measurement errors are present in both the left-hand side variable $G_t$ and the right-hand side variable $E_t$ of the regression \eqref{eq:AFnew}. It is well-known that measurement errors may lead to biases in parameters estimates \citep[e.g.][]{Klepper1984}. As a robustness check, we run a so-called \emph{Deming regression} \citep[][]{Deming1943} which is a regression that takes potential measurement errors into account when formulating the regression-based estimator.

Formally, we assume that the data $G_t$ and $E_t$ are noisy measurements of the variables $G_t^\ast$ and $E_t^\ast$, where
\begin{align}\label{eq:demingModel}
G_t^\ast = \alpha E_t^\ast + u_t
\end{align}
and
\begin{align*}
G_t &=  G_t^\ast + \epsilon_{G,t}, \quad \epsilon_{G,t} \stackrel{iid}{\sim} N(0,\sigma_{G,1}^2), \\
E_t &=  E_t^\ast + \epsilon_{E,t}, \quad \epsilon_{E,t} \stackrel{iid}{\sim} N(0,\sigma_{E,1}^2), 
\end{align*}
with $\epsilon_{G,t}$ and $\epsilon_{E,t}$ being statistically independent. Let $\delta = \sigma_{G,1}^2 / \sigma_{E,1}^2$ denote the ratio of the two measurement error variances. Then the Deming regression estimate of $\alpha$ in \eqref{eq:demingModel} using the data $G_t$ and $E_t$ is
\begin{align}\label{eq:deming}
\hat \alpha_{deming} = \frac{M_{GG} - \delta M_{EE} + \sqrt{ (M_{GG} - \delta M_{EE})^2 + 4 \delta M_{EG}^2}}{2 M_{EG}},
\end{align}
where
\begin{align*}
M_{GG}  &= \frac{1}{T} \sum_{t=1}^T G_t^2, \\
M_{EE}  &= \frac{1}{T} \sum_{t=1}^T E_t^2, \\
M_{EG}  &= \frac{1}{T} \sum_{t=1}^T E_t G_t.
\end{align*}

We apply the Deming estimator \eqref{eq:deming} to the three data sets (GCP, H\&C, vMa) and two subsamples (1959--2022 and 1992--2022). Since the measurement error ratio $\delta$ is unknown, we run the regression for several values of $\delta$, namely $\delta \in \{0.2, 0.5, 1, 2, 5\}$. The results are shown in Table \ref{tab:deming}. Comparing with the estimates $\hat \alpha_2$ obtained from Equation \eqref{eq:AFnew} and reported in Tables \ref{tab:analysis_all} and \ref{tab:analysis1992}, we see that 
although the Deming estimates are slightly larger than those obtained from the regression-based estimator $\hat \alpha_2$,  the Deming estimates are  always within the confidence bands of the regression-based estimator $\hat \alpha_2$, suggesting that potential measurement error is not driving the results reported in this study.

\section{Cumulative airborne fraction}\label{app:caf}

Another
alternative to the ratio-based and regression-based estimators of the AF is the \emph{cumulative airborne fraction} (CAF) estimator,
defined as the ratio of the cumulative  changes in atmospheric CO$_2$ concentrations and
the cumulative CO$_2$ emissions, given by
\begin{align}\label{eq:CAFapp}
\text{CAF}_{t_0,t} = \sum _{j=t_0}^t G_j \, \Large{/} \, \sum _{j=t_0}^t E_j,
\end{align}
where $t_0$ denotes the first year of the sample and $t$ is the current year or the last year of the sample.
Using the Global Carbon Project data $1959$--$2022$, we can calculate the CAF, yielding
\[
\text{CAF}_{1959,2022}^{\text{GCP}} = \sum _{j=1959}^{2022} G_j \, \Large{/} \, \sum _{j=1959}^{2022} E_j = \frac{218.03}{491.03} = 44.40\%.
\]
The two alternative data sets, introduced in Section \ref{app:diffdata}, result in
\[
\text{CAF}_{1959,2022}^{\text{H\&C}} = \sum _{j=1959}^{2022} G_j \, \Large{/} \, \sum _{j=1959}^{2022} E_j = \frac{218.03}{457.65} = 47.64\%, 
\]
and
\[
\text{CAF}_{1959,2022}^{\text{vMa}} = \sum _{j=1959}^{2022} G_j \, \Large{/} \, \sum _{j=1959}^{2022} E_j = \frac{218.03}{442.07} = 49.32\%.
\]

Due to its cumulative nature, the CAF is not as useful for studying changes and trends
in the AF, when compared to the ratio- and regression-based measures considered in the
main paper, which utilize the yearly variables $G_t$ and $E_t$.
To illustrate this point, Figure \ref{fig:CAFapp} shows the CAF over the period 2023--2100
using the data from the SSP1-2.6 scenario, i.e. the same data as used in
Figure \ref{fig:SSP126} of the main paper. Panel a) of the figure shows the $CAF_{2023,t}$ $t = 2023, 2024, \ldots, 2100$, as defined in Equation \eqref{eq:CAFapp}.
As can be seen, the CAF is very smooth and generally adjusts slowly to changes in
the behavior of the carbon system. 

A possible remedy to the ``over-smoothness''
is to calculate the CAF over a (smaller) fixed window of length $w \geq 1$.
Consider therefore the alternative CAF, defined as
\begin{align}\label{eq:CAFapp2}
\text{CAF}_{t_0,t}(w) = \text{CAF}_{\max\{t_0-w+1,t_0\},t} = \sum _{j= \max\{t_0-w+1,t_0\}}^t G_j \, \Large{/} \, \sum _{j=\max\{t_0-w+1,t_0\}}^t E_j ,
\end{align}
for a window length $w \geq 1$.
This is the ``moving-window'' version of the CAF, where each data point is calculated
using only the preceding $w$ years (when $t-t_0+1<w$, then $t-t_0+1$ years are used).
For instance, $w=1$ corresponds to the conventional ratio-based AF as considered in
the main paper, while $w = 78$ is the CAF suggested in Equation \eqref{eq:CAFapp}
above (because there are $78$ yearly data points in the 2023--2100 SSP data).
For other values of $w$, but with $1 < w < 78$, the $\text{CAF}(w)$ in
\eqref{eq:CAFapp2} is calculated using the cumulative data using shorter windows.
This resulting CAF may be able to track the changes and trends in the AF better. 

If $w$ is chosen too large, the CAF in \eqref{eq:CAFapp2} will inherit the
downsides of the ``full-sample'' CAF in \eqref{eq:CAFapp}, i.e. it will be slow to adapt.
Conversely, in case $w$ is too small, 
this version of the CAF inherits the downsides of the ratio-based estimator considered
in the main paper since it is still defined as a ratio.
To illustrate these points, Panels b)--f) of Figure \ref{fig:CAFapp}
show $\text{CAF}_{2023,t}(w)$,  $t = 2023, 2024, \ldots, 2100$, for $w = 50,20,10,5,1$,
respectively. In particular,
when
choosing a smaller window size $w$,
the CAF is indeed tracking the changes in the behavior of the carbon system more closely,
while it also has
the effect of making the CAF vulnerable to periods
where $\sum E_t \approx 0$.




%
%

\section{Simulation study: Performance of the estimators in conditions similar to the data sample $1959$--$2022$}\label{app:simstudy} 
We simulate data from the models in \eqref{eq:AForig}--\eqref{eq:AFnew}. To facilitate comparison, we set the data-generating value of $\alpha$ to $0.4386$, which is the estimated value obtained using the model in \eqref{eq:AForig} and the GCP data set, see Table \ref{tab:analysis}. The variance of the error terms are set to $Var(\epsilon_t^{(1)}) = 0.1258^2$ and $Var(u_t^{(2)}) = 0.9088^2$, again corresponding to the estimates from the GCP data (Table \ref{tab:analysis}). Using these values, we simulate data $y_t = G_t/E_t$ from model \eqref{eq:AForig} directly. To simulate from model \eqref{eq:AFnew}, we first need to simulate an emissions trajectory. To this end, we specify emissions as a random walk plus drift, i.e.  $E_t = E_{t-1} + b + \xi_t$, where $b$ is drift parameter, as motivated in this study. 

Using the GCP data for $1959$--$2022$, we estimate $b$ and $Var(\xi_t)$ from the model $\Delta E_t := E_t -E_{t-1} = b + \xi_t$. We also set $E_0 = 4.3433$, which is the first observation of $E_t$ from the GCP data. Then, we simulate emissions time series  using the dynamic equation $E_t  = E_{t-1} + \hat b +\xi_t$, $t = 1,2,\ldots, T$, where the variance of $\xi_t$ is set to the value $\widehat{Var}(\xi_t)$ found on the GCP data.

We simulate $10^6$ instances of $G_t$ and $E_t$ for $t = 1,2,\ldots, T$. For each instance, we estimate $\alpha$ using the estimators $\hat \alpha_1$ and $\hat \alpha_2$, and calculate the squared errors of the resulting estimates. The root mean squared errors (RMSE) of the two estimators are obtained by averaging the squared errors over the $10^6$ simulations. We run these simulation studies for various sample sizes, namely $T = 64, 65, \ldots, 142$. $T = 64$ corresponds to the sample size relevant for the data in the main paper (1959--2022), while $T = 142$ corresponds to data until $2100$ (1959--2100). The RMSEs of the two estimators, as functions of sample size $T$, are shown in the left plot of Figure \ref{fig:simstudy}. The relative RMSE, i.e. RMSE($\hat \alpha_2$)$/$RMSE($\hat \alpha_1$), as a function of $T$, is shown in the right plot of Figure \ref{fig:simstudy}.

From the left plot of Figure \ref{fig:simstudy}, we see that the RMSE of both estimators are decreasing in $T$. The decrease is much more pronounced for the estimator $\hat \alpha_2$ compared to the estimator $\hat \alpha_1$. This is to be expected, since we saw above that $\hat \alpha_2$ converges to $\alpha$ at the rate $T^{3/2}$, while $\hat \alpha_1$ converges at the slower rate $T^{1-\gamma}$, $\gamma>0$. The difference between the performance of the two estimators is illustrated in the right plot of Figure \ref{fig:simstudy},  where the relative RMSE is shown. For $T = 64$, corresponding to the leftmost observation, the RMSE of $\hat \alpha_2$ is approximately $10\%$ lower than the RMSE of $\hat \alpha_1$. These simulation results align with the empirical estimates, where it was found that $\hat \alpha_2$ was approximately $10\%$ more precise than $\hat \alpha_1$ (Tables \ref{tab:analysis} and \ref{tab:analysis_all}). The relative performance continues to diverge as the sample size gets larger; for $T = 142$, the reduction in RMSE is on the order of $50\%$.

\section{Details on the Kalman filter approach to estimating the time-varying airborne fraction}\label{app:KF}
Consider the state-space model
\begin{align*}
G_t &= \alpha_t E_t + u_t^{(2)}, \quad u_t^{(2)} \stackrel{iid}{\sim}N(0,\sigma_u^2),  \\
\alpha_{t+1} &=
    \begin{cases}
      \alpha_{t} + \eta_{t} & \text{ if } t+1 \neq \tau,\\
      1-\alpha_{t} + \eta_{t} & \text{ if } t+1 = \tau,
    \end{cases} , \quad \eta_t\stackrel{iid}{\sim}N(0,\sigma_\eta^2),
\end{align*}
where $E_t$ is treated as a covariate (exogenous variable) and  $\tau$ denotes the year where emissions first turn negative, i.e. $\tau = \inf \{ t: E_t<0\}$. (We use the convention that if $  \{ t: E_t<0\}$ is empty, then $\tau = \infty$.)
This is a random walk model for $\alpha_t$, i.e. $\alpha_{t+1} = \alpha_t + \eta_t$, for all years $t$ except at the year $t+1 = \tau$, where the transition equation becomes $\alpha_{t+1} =  1-\alpha_{t} + \eta_{t}$.\footnote{The  reason for specifying the model in a way that reflects $\alpha_t$ around one at $t=\tau$ is as follows: When $G_t<0$ and $E_t>0$ we expect that  $\alpha_{t}<0$, which means that the carbon sinks absorb more than is emitted. When emissions turn negative, i.e. $E_t<0$, then $\alpha_t$ switches to being positive. Because we still expect that the sinks absorb more than what is emitted, we would in fact expect that $\alpha_t>1$. This features is obtained by the specification $\alpha_{\tau} =  1-\alpha_{\tau-1} + \eta_{\tau}$.}

Given the data for $G_t$, we may use the Kalman filter to evaluate the likelihood of the model. 
The parameters $\sigma_u^2$ and $\sigma_\eta^2$ may be estimated by the method of maximum likelihood, and the Kalman smoother provides an estimate of the time-varying parameter $\alpha_t$, namely $\hat \alpha_{2,t} = \widehat\E(\alpha_t | \text{data})$, where ``data'' denotes $(G_t,E_t)$ for $t =  2023, 2024, \ldots, 2100$. Similarly, the Kalman smoother provides confidence bands around $\hat \alpha_{2,t}$. We refer to \cite{Durbin2012} for a textbook treatment of state space models and the Kalman filter and smoother.


\subsection{Application to SSP scenarios 2023-2100}\label{app:SSP}
Figures \ref{fig:SSP119}--\ref{fig:SSP585} present the results from applying the regression based estimator $\hat \alpha_{2,t} = \widehat\E(\alpha_t | \text{data})$ of the time-varying airborne fraction $\alpha_t$, obtained from the Kalman smoother, to scenarios SSP1-1.9, SSP2-4.5, SSP3-7.0, SSP4-3.4, and SSP5-8.5, respectively. The SSP scenarios were run using the simple climate model MAGICC \citep{MAGICC}.\footnote{The SSP scenarios can be run in MAGICC in a web browser via the link \url{https://live.magicc.org/scenarios/bced417f-0f7f-4bb7-8359-792a0a8b0368/overview}. Last accessed June 17, 2024.}

\newpage

\clearpage

\begin{table}[h!]
\caption{\it $p$-values for the Augmented Dickey-Fuller (ADF) test for a unit root and the Engle-Granger test for cointegration between $G_t$ and $E_t$ on the GCP data (1959--2022), see Section \ref{app:unitroot} for details. First six rows are ADF tests and the last row is the Engle-Granger test. ``AR'': The null model is a unit root model while the alternative model is an AR model. ``ARD'': The null model is a unit root model while the alternative model is an AR model with drift. ``TS'': The null model is a unit root model with a constant (i.e. a drift in levels) while the alternative model is an AR model with a time trend. ``Engle-Granger'': The null hypothesis is that the error term $u_t$ in $G_t = \alpha E_t + u_t$ has a unit root, i.e. that $G_t$ and $E_t$  are not cointegrated. The value ``$L$''  denotes the number of lags included in both the null and alternative models, see equations \eqref{eq:adf null}--\eqref{eq:adf alternative}.}
\vspace{4mm}
\centering
\scriptsize
\begin{tabular}{lcccccc@{}} \hline
 & $L=0$ & $L=1$ & $L=2$ &  $L=3$  & $L=4$ & $L=5$   \\ \hline 
$y_t = G_t$ (AR)   & 0.2349  &  0.4867    &0.6597 &   0.7982 &   0.8403&    0.8470 \\
$y_t = G_t$ (ARD)  &  0.0039  &  0.0778  &  0.3105 &   0.4384    &0.4564 &   0.3895 \\
$y_t = G_t$ (TS)   &  0.0010  &  0.0010&    0.0035 &   0.0192&    0.0231 &   0.0056 \\
$y_t = E_t$ (AR)   & 0.9990 &   0.9990   & 0.9990  &  0.9990 &   0.9953  &  0.9984 \\
$y_t = E_t$ (ARD)  & 0.9488   & 0.9329  &  0.9038  &  0.8533  &  0.8201 &   0.8871 \\
$y_t = E_t$ (TS)  & 0.0933   & 0.2108 &   0.3108  &  0.3339 &   0.2250   & 0.4775 \\
Engle-Granger &      0.0010 &   0.0010   & 0.0019 &   0.0135 &   0.0173   & 0.0038 \\ \hline
\end{tabular}
\label{tab:unitroot}
\end{table}

\newpage

\clearpage

\begin{figure}[t!]
    \centering
    \caption{Data used in the study over the period $1959$--$2022$. The dashed vertical black line denotes the year $1992$, the starting point for the analysis in Section \ref{app:analysis1992}. }
    \includegraphics[width=0.99\textwidth]{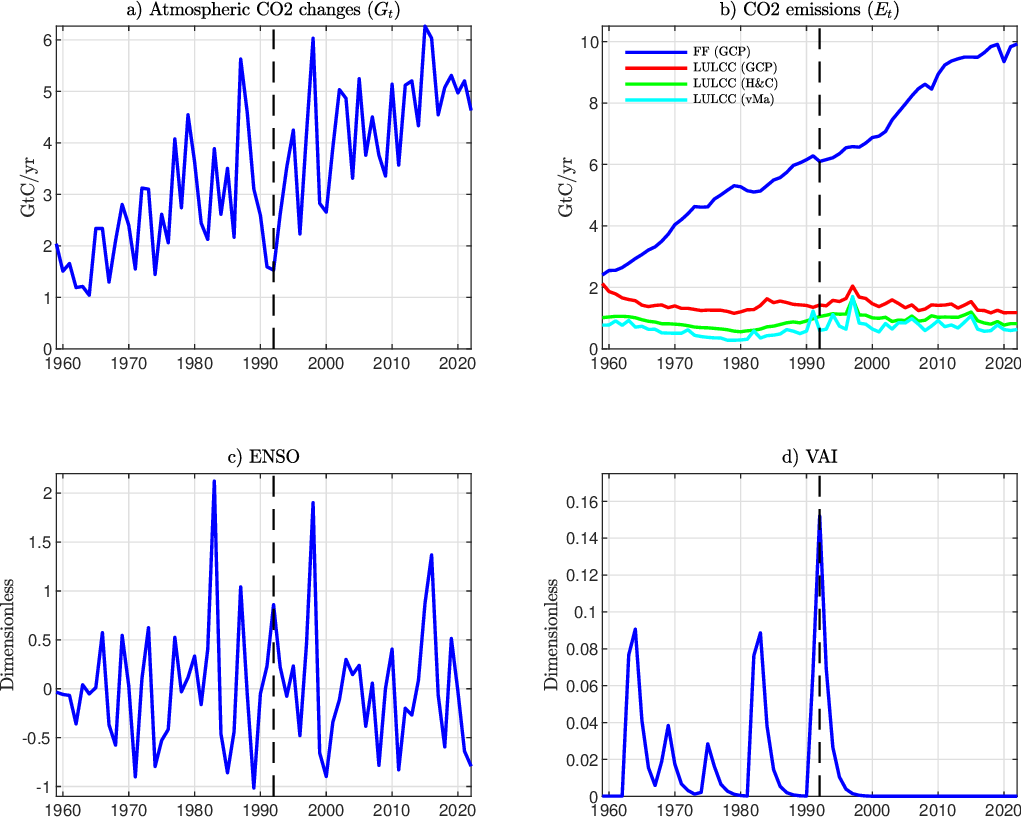}
    \label{fig:data_all}
\end{figure}

\newpage

\clearpage

\begin{table}[h!]
\caption{\it Least-squares regression output of the linear models in Equations \eqref{eq:AForig}--\eqref{eq:AFnew2} for the period 1959--2022 for the three data sets GCP, H\&C, and vMa. ``$\hat \alpha$'' denotes the estimate of the AF $\alpha$,   ``SE($\hat \alpha$)'' denotes the standard error of $\hat \alpha$, and ``Relative SE'' denotes the ratio of the standard error to the standard error from $\hat \alpha_1$, corresponding to the leftmost column. ``$CI_{95\%}(\alpha)$'' denotes a $95\%$ confidence interval for the AF $\alpha$, based on the Gaussian distribution. ``$\widehat{SD}(u_t)$'' denotes the estimated standard deviation of the error term $u_t$ and $R^2$ is the coefficient of determination. ``$\hat \gamma_i$'', $i=1,2$, denotes the estimates of the $\gamma_i$ parameter from Equations \eqref{eq:AForig2}--\eqref{eq:AFnew2}  and ``SD($\hat \gamma_i$)'' is their estimated standard deviation.}
\vspace{4mm}
\centering
\scriptsize
\begin{tabular}{lcccc@{}} \hline
Data: GCP &  \eqref{eq:AForig} &  \eqref{eq:AFnew} &  \eqref{eq:AForig2} &  \eqref{eq:AFnew2}   \\ \hline 
$\hat \alpha$    &                0.4386   & 0.4478    &0.4716    &0.4697 \\
SE($\hat \alpha$)   &             0.0159&    0.0141   & 0.0126  &  0.0105  \\
Relative SE   &     1.0000  &  0.8895   & 0.7904   & 0.6630 \\
 $CI_{95\%}(\alpha)$  &  $ [ 0.4074, 0.4697 ]$ &    $ [0.4201, 0.4755 ]$ &    $ [0.4470, 0.4962]$  &    $ [0.4490, 0.4903]$\\ 
$\widehat{SD}(u_t)$  &             0.1258   & 0.9088   & 0.0881   & 0.6292 \\
$R^2$ &       0  &  0.5863 &   0.5258   & 0.8080 \\ 
 $\hat \gamma_1$ (ENSO)   &  -   &    -   &0.1365  &  1.0025 \\ 
 SD($\hat \gamma_1$)   &      -  &     - &    0.0168 &   0.1141\\ 
  $\hat \gamma_2$  (VAI) &   -    &   -  &  -2.2676 & -15.1482 \\ 
 SD($\hat \gamma_2$)     &      -   &    -   &  0.2935 &   2.3688\\ \hline
 
Data: H\&C &  \eqref{eq:AForig} &  \eqref{eq:AFnew} &  \eqref{eq:AForig2} &  \eqref{eq:AFnew2}  \\ \hline 
$\hat \alpha$    &          0.4763   &   0.4755   &   0.5102   &   0.4960  \\
SE($\hat \alpha$)   &             0.0171   &   0.0151    &  0.0144   &   0.0117 \\
Relative SE    &              1.0000  &    0.8800  &    0.8394   &   0.6814 \\
 $CI_{95\%}(\alpha)$    &        $[ 0.4427 ,0.5099]$    &        $0.4460 ,0.5051]$    &        $[ 0.4820  ,0.5384]$    &        $[ 0.4731, 0.5189]$   \\ 
$\widehat{SD}(u_t)$  &    0.1371  &    0.9148  &    0.0988   &   0.6648 \\
$R^2$ &              0   &   0.5807    &  0.4975    &  0.7856 \\    
   $\hat \gamma_1$ (ENSO)     &     -   &     -   &  0.1468   &   0.9851 \\ 
 SD($\hat \gamma_1$)    &       -  &      - &  0.0205    &  0.1279 \\ 
  $\hat \gamma_2$  (VAI)    &    -     &   - & -2.3289 &   -13.6909 \\ 
 SD($\hat \gamma_2$)       &        -  &      -  &    0.3987    &  3.2461 \\ \hline
 
Data: vMa &  \eqref{eq:AForig} &  \eqref{eq:AFnew} &  \eqref{eq:AForig2}  &  \eqref{eq:AFnew2}  \\ \hline 
$\hat \alpha$    &          0.4957     &  0.4902    &   0.5298     &  0.5096 \\
SE($\hat \alpha$)   &        0.0180      & 0.0155    &   0.0158     &  0.0124  \\
Relative SE      & 1.0000      & 0.8628      & 0.8822      & 0.6935 \\
 $CI_{95\%}(\alpha)$   &  $[0.4605, 0.5309]$    &  $[ 0.4599, 0.5206 ]$    &  $[ 0.4988, 0.5609 ]$    &  $[ 0.4852, 0.5340]$ \\ 
$\widehat{SD}(u_t)$  &           0.1431    &   0.9155      & 0.1059       &0.6872 \\
$R^2$ &     0      & 0.5801     &  0.4704      & 0.7709 \\ 
 $\hat \gamma_1$    (ENSO)     &   - &       -   &    0.1496     &  0.9564 \\ 
 SD($\hat \gamma_1$)       &     -      &  -  &    0.0220   &    0.1328 \\ 
  $\hat \gamma_2$    (VAI)   &   - &       - &    -2.3440    & -12.9004 \\ 
 SD($\hat \gamma_2$)           &   -      &  -   & 0.4594   &    3.5267 \\ \hline
\end{tabular}
\label{tab:analysis_all}
\end{table}

\newpage

\clearpage

\begin{table}[h!]
\caption{\it Least-squares regression output of the linear models in Equations \eqref{eq:AForig}--\eqref{eq:AFnew2} for the period 1992--2022 for the three data sets GCP, H\&C, and vMa. ``$\hat \alpha$'' denotes the estimate of the AF $\alpha$,   ``SE($\hat \alpha$)'' denotes the standard error of $\hat \alpha$, and ``Relative SE'' denotes the ratio of the standard error to the standard error from $\hat \alpha_1$, corresponding to the leftmost column. ``$CI_{95\%}(\alpha)$'' denotes a $95\%$ confidence interval for the AF $\alpha$, based on the Gaussian distribution. ``$\widehat{SD}(u_t)$'' denotes the estimated standard deviation of the error term $u_t$ and $R^2$ is the coefficient of determination. ``$\hat \gamma_i$'', $i=1,2$, denotes the estimates of the $\gamma_i$ parameter from Equations \eqref{eq:AForig2}--\eqref{eq:AFnew2}  and ``SD($\hat \gamma_i$)'' is their estimated standard deviation.}
\vspace{4mm}
\centering
\scriptsize
\begin{tabular}{lcccc@{}} \hline
Data: GCP &  \eqref{eq:AForig} &  \eqref{eq:AFnew} &  \eqref{eq:AForig2} &  \eqref{eq:AFnew2}   \\ \hline 
$\hat \alpha$    &                   0.4456   & 0.4497   & 0.4626  &  0.4613 \\
SE($\hat \alpha$)   &        0.0190  &  0.0157   & 0.0124  &  0.0104 \\
Relative SE   &      1.0000   & 0.8247   & 0.6509   & 0.5464  \\
 $CI_{95\%}(\alpha)$  &  $ [0.4083 ,0.4828 ]$    &  $ [0.4190 ,0.4804]$    &  $ [0.4384 ,0.4869  ]$    &  $ [0.4409,   0.4816]$ \\ 
$\widehat{SD}(u_t)$  &          0.1065 &   0.9309 &   0.0662  &  0.5891  \\
$R^2$   &    0  &  0.3558  &  0.6391   & 0.7592\\ 
 $\hat \gamma_1$    (ENSO)     &   - &       -   &    0.1122    1.0221 \\ 
 SD($\hat \gamma_1$)       &     -      &  -  &   0.0166    0.1186 \\ 
  $\hat \gamma_2$    (VAI)   &   - &       - &   -2.2281  -17.5878 \\ 
 SD($\hat \gamma_2$)           &   -      &  -   &  0.1702    1.2190 \\ \hline
 
Data: H\&C &  \eqref{eq:AForig} &  \eqref{eq:AFnew} &  \eqref{eq:AForig2} &  \eqref{eq:AFnew2}  \\ \hline 
$\hat \alpha$    &            0.4662   & 0.4691   & 0.4841  &  0.4811 \\
SE($\hat \alpha$)   &       0.0203  &  0.0163 &   0.0136   & 0.0111  \\
Relative SE    &          1.0000   & 0.8041    &0.6678   & 0.5478   \\
 $CI_{95\%}(\alpha)$    &        $[ 0.4264, 0.5060  ]$    &        $[0.4371 ,0.5011 ]$    &        $[ 0.4575 ,0.5107]$    &        $[ 0.4593 ,0.5029]$\\ 
$\widehat{SD}(u_t)$  &          0.1132  &  0.9346  &  0.0706 &   0.5921 \\ 
$R^2$     &        0   & 0.3506  &  0.6372   & 0.7567\\
 $\hat \gamma_1$    (ENSO)     &   - &       -   & 0.1203   & 1.0333 \\ 
 SD($\hat \gamma_1$)       &     -      &  -  &     0.0191  &  0.1238\\ 
  $\hat \gamma_2$    (VAI)   &   - &       - &    -2.3345&  -17.4042\\ 
 SD($\hat \gamma_2$)           &   -      &  -   &   0.2018  &  1.3348 \\ \hline

Data: vMa &  \eqref{eq:AForig} &  \eqref{eq:AFnew} &  \eqref{eq:AForig2}  &  \eqref{eq:AFnew2}  \\ \hline 
$\hat \alpha$    &        0.4786     &  0.4815    &  0.4961  &    0.4924  \\
SE($\hat \alpha$)   &        0.0203  &    0.0159   &   0.0141  &    0.0113 \\
Relative SE      &          1.0000   &   0.7865    &  0.6962    &  0.5598  \\
 $CI_{95\%}(\alpha)$   &  $ [   0.4389 ,0.5183 ]$    &  $ [ 0.4503 ,0.5128  ]$    &  $ [ 0.4685,  0.5238]$    &  $ [0.4702 ,0.5146]$ \\ 
$\widehat{SD}(u_t)$  &            0.1146   &   0.9109    &  0.0732  &    0.5914 \\ 
$R^2$     &            0    &  0.3832    &  0.6198   &   0.7573 \\ 
 $\hat \gamma_1$    (ENSO)     &   - &       -   &   0.1213   &   1.0060\\ 
 SD($\hat \gamma_1$)       &     -      &  -  &    0.0207    &  0.1272 \\ 
  $\hat \gamma_2$    (VAI)   &   - &       - &    -2.2999   & -16.1863 \\ 
 SD($\hat \gamma_2$)           &   -      &  -   &  0.2268  &    1.3944 \\ \hline
\end{tabular}
\label{tab:analysis1992}
\end{table}

\newpage

\clearpage

    \begin{table}[t!]

\caption{\it Deming regression estimates \eqref{eq:deming} of the AF $\alpha$ in \eqref{eq:demingModel} for the three data sets GCP, H\&C, and vMa, for the full sample 1959--2022 (left panel) and the recent sample 1992--2022 (right panel). The parameter $\delta$ denotes the measurement error ratio $\delta$, see Section \ref{app:deming} for details.}
\vspace{4mm}
\centering
\scriptsize
\begin{tabular}{lccccc|ccccc@{}}
& \multicolumn{5} {c} {\bfseries Full sample (1959-2022)} & \multicolumn{5} {c} {\bfseries Recent sample (1992-2022)} \\
\cmidrule{2-11}
Data   &  $\delta = 0.2$ & $\delta = 0.5$ & $\delta = 1$ & $\delta = 2$ & $\delta = 5$  &  $\delta = 0.2$ & $\delta = 0.5$ & $\delta = 1$ & $\delta = 2$ & $\delta = 5$  \\ \hline 
GCP  &  0.4623 &    0.4561  &   0.4526  &   0.4504  &   0.4489  &  0.4598  &   0.4555  &   0.4531  &   0.4515  &   0.4504  \\
H\&C &     0.4921 &    0.4852   &  0.4813   &  0.4787   &  0.4769    &  0.4802  &   0.4756   &  0.4729   &  0.4712 &    0.4700\\
vMa  &  0.5078 &    0.5007 &    0.4964   &  0.4936&     0.4917    &  0.4926 &    0.4881  &   0.4854  &   0.4836 &    0.4824 \\ \hline
\end{tabular}
\label{tab:deming}
\end{table}

\newpage

\clearpage

\begin{figure}[t!]
    \centering
    \caption{Cumulative airborne fraction \eqref{eq:CAFapp2} calculated from SSP1-2.6 data over the period $2023$--$2100$. a) $w = 78$, corresponding to the CAF in \eqref{eq:CAFapp}; b) $w= 50$; c) $w = 20$; d) $w=10$; e) $w = 5$; f) $w = 1$, corresponding to the AF ratio-based estimator, denoted by $\widehat \alpha_{1,t}$ in the main paper. See Section \ref{app:caf} for details. \\}
    \includegraphics[width=0.99\textwidth]{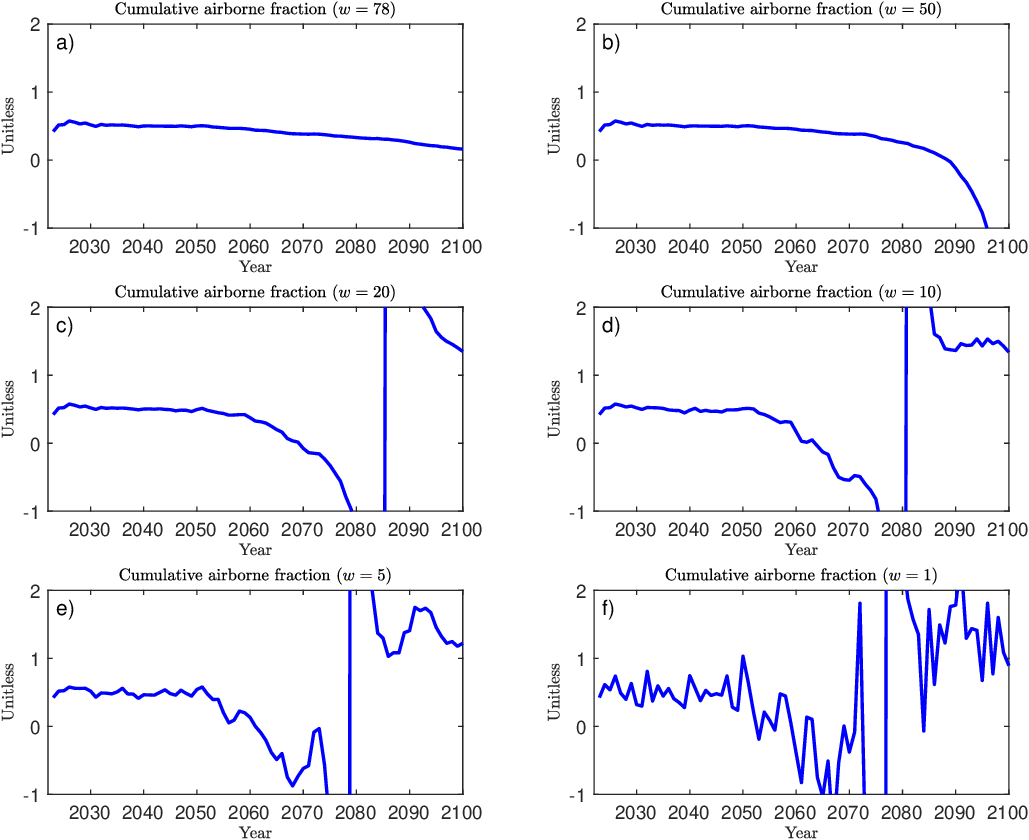}
    \label{fig:CAFapp}
\end{figure}

\newpage

\clearpage

\begin{figure}[t!]
    \centering
    \caption{Simulation study of the properties of the estimators $\hat \alpha_1$ and $\hat \alpha_2$. See Section \ref{app:simstudy} for details. }
    \includegraphics[width=0.99\textwidth]{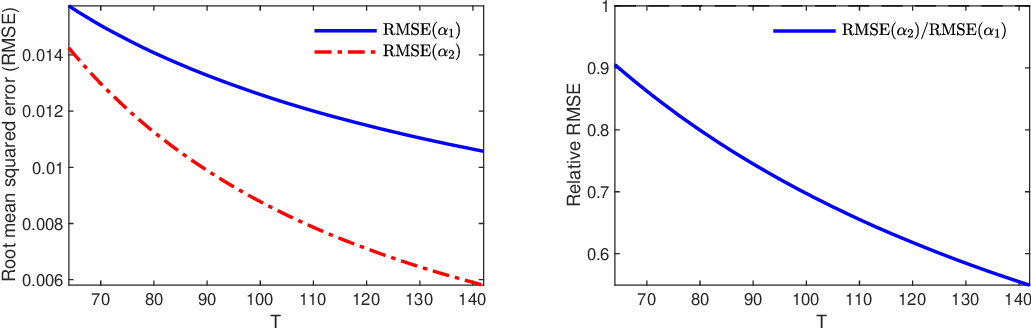}
    \label{fig:simstudy}
\end{figure}

\newpage

\clearpage

\begin{figure}[t!]
    \centering
   \caption{\emph{Analysis of SSP1-1.9 data over the period $2023$--$2100$. a) Atmospheric concentration changes ($G_t$) for the historical period 1959--2022 (black) and the SSP period 2022--2100.  b) Emissions ($E_t$) data. Magenta lines denote the original output data from MAGICC and the blue lines denote the perturbed data. c) Ratio of atmospheric changes to emissions ($G_t/E_t$). Black lines denote data from the Global Carbon Project over the historical period $1959$--$2022$.  Magenta lines denote the original SSP output data from MAGICC over the period $2023$--$2100$. Blue lines denote the perturbed SSP data. The red line in c) is the regression-based estimator $\hat \alpha_{2,t} = \widehat\E(\alpha_t | \text{data})$ of the time-varying airborne fraction $\alpha_t$, obtained from the Kalman smoother. Shaded area is a $95\%$ confidence band around $\hat \alpha_{2,t}$.} }
    \includegraphics[width=0.99\textwidth]{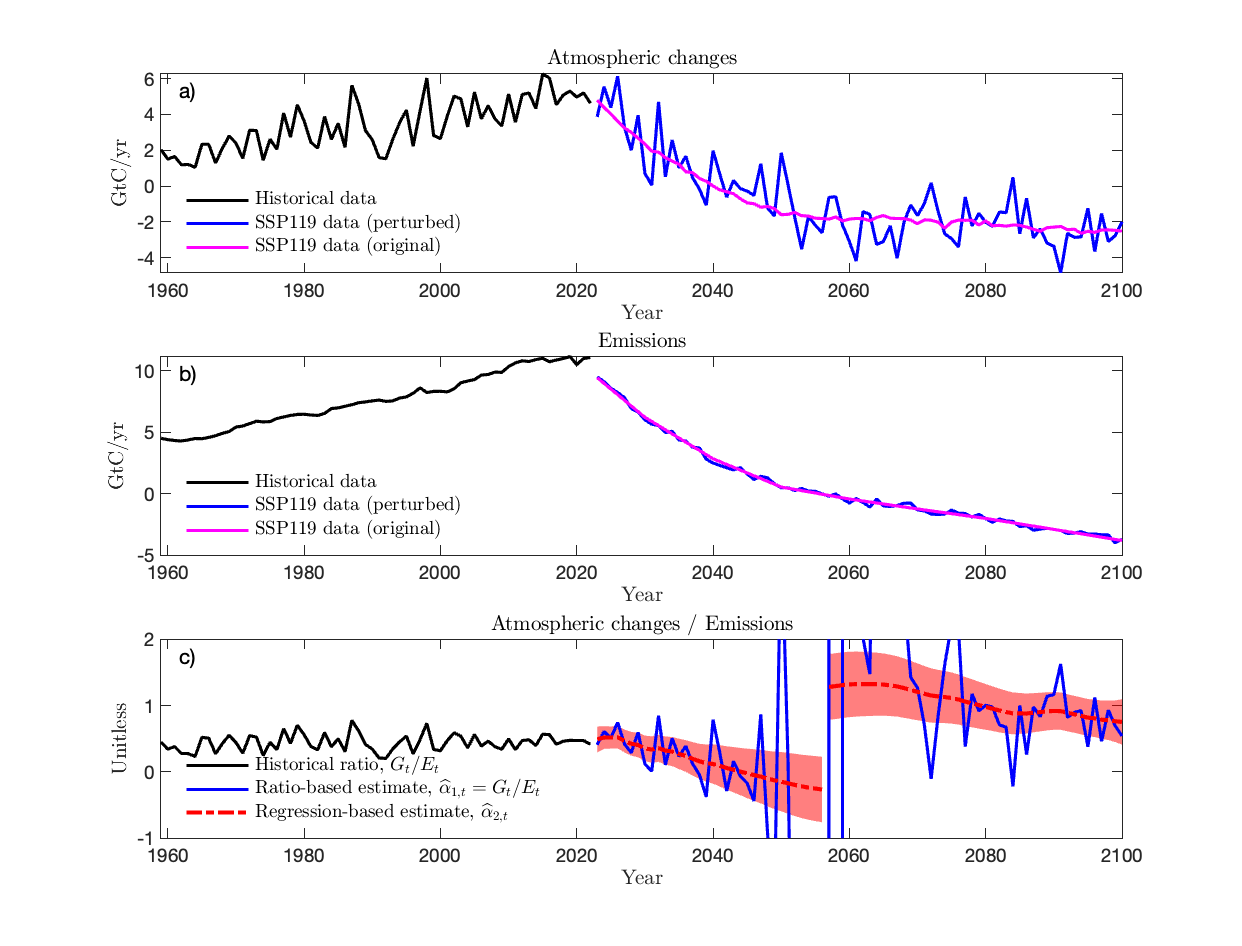}
    \label{fig:SSP119}
\end{figure}

\begin{figure}[t!]
    \centering
   \caption{\emph{Analysis of SSP2-4.5 data over the period $2023$--$2100$. a) Atmospheric concentration changes ($G_t$) for the historical period 1959--2022 (black) and the SSP period 2022--2100.  b) Emissions ($E_t$) data. Magenta lines denote the original output data from MAGICC and the blue lines denote the perturbed data. c) Ratio of atmospheric changes to emissions ($G_t/E_t$). Black lines denote data from the Global Carbon Project over the historical period $1959$--$2022$.  Magenta lines denote the original SSP output data from MAGICC over the period $2023$--$2100$. Blue lines denote the perturbed SSP data. The red line in c) is the regression-based estimator $\hat \alpha_{2,t} = \widehat\E(\alpha_t | \text{data})$ of the time-varying airborne fraction $\alpha_t$, obtained from the Kalman smoother. Shaded area is a $95\%$ confidence band around $\hat \alpha_{2,t}$.} }
    \includegraphics[width=0.99\textwidth]{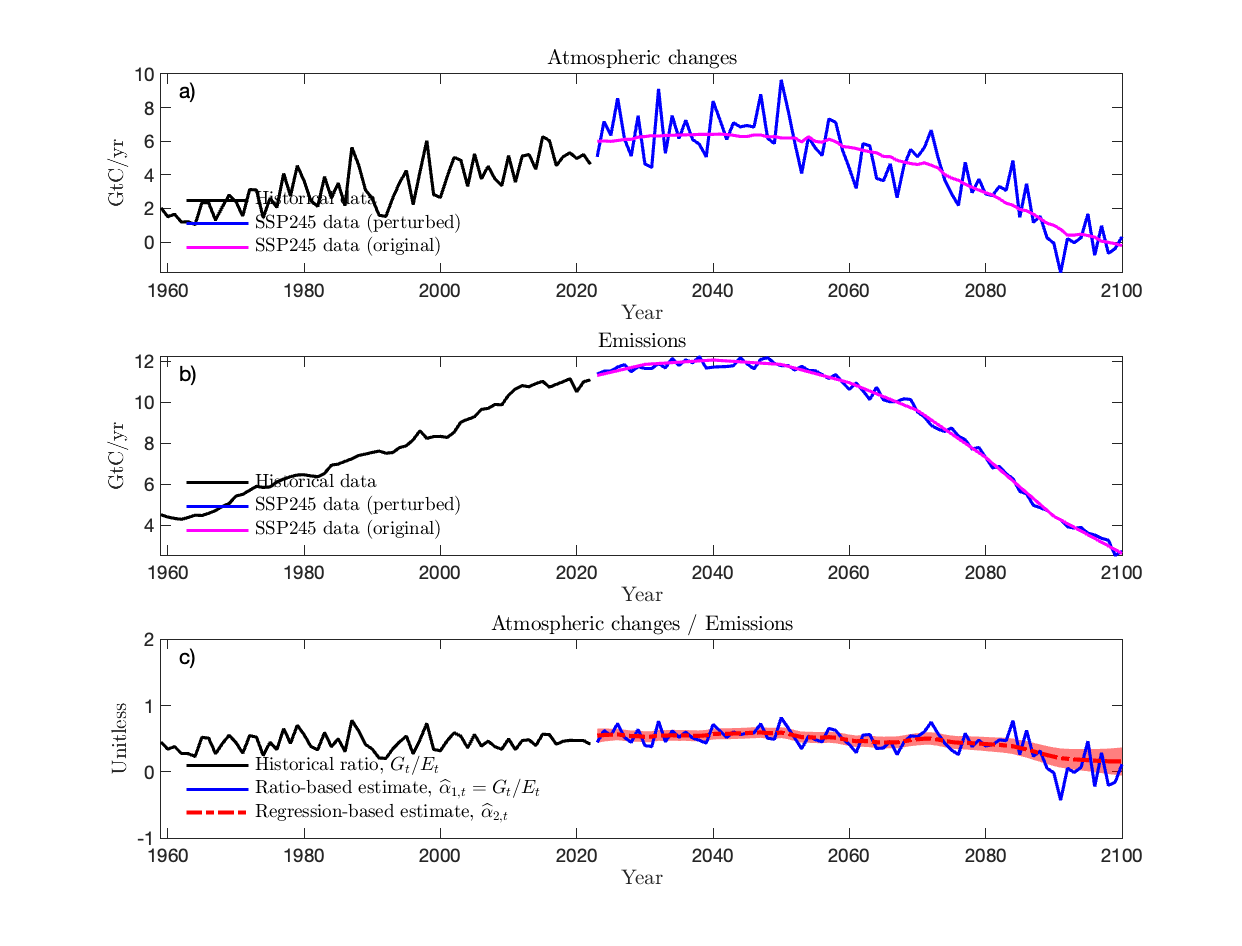}
    \label{fig:SSP245}
\end{figure}

\begin{figure}[t!]
    \centering
   \caption{\emph{Analysis of SSP3-7.0 data over the period $2023$--$2100$. a) Atmospheric concentration changes ($G_t$) for the historical period 1959--2022 (black) and the SSP period 2022--2100.  b) Emissions ($E_t$) data. Magenta lines denote the original output data from MAGICC and the blue lines denote the perturbed data. c) Ratio of atmospheric changes to emissions ($G_t/E_t$). Black lines denote data from the Global Carbon Project over the historical period $1959$--$2022$.  Magenta lines denote the original SSP output data from MAGICC over the period $2023$--$2100$. Blue lines denote the perturbed SSP data. The red line in c) is the regression-based estimator $\hat \alpha_{2,t} = \widehat\E(\alpha_t | \text{data})$ of the time-varying airborne fraction $\alpha_t$, obtained from the Kalman smoother. Shaded area is a $95\%$ confidence band around $\hat \alpha_{2,t}$.} }
    \includegraphics[width=0.99\textwidth]{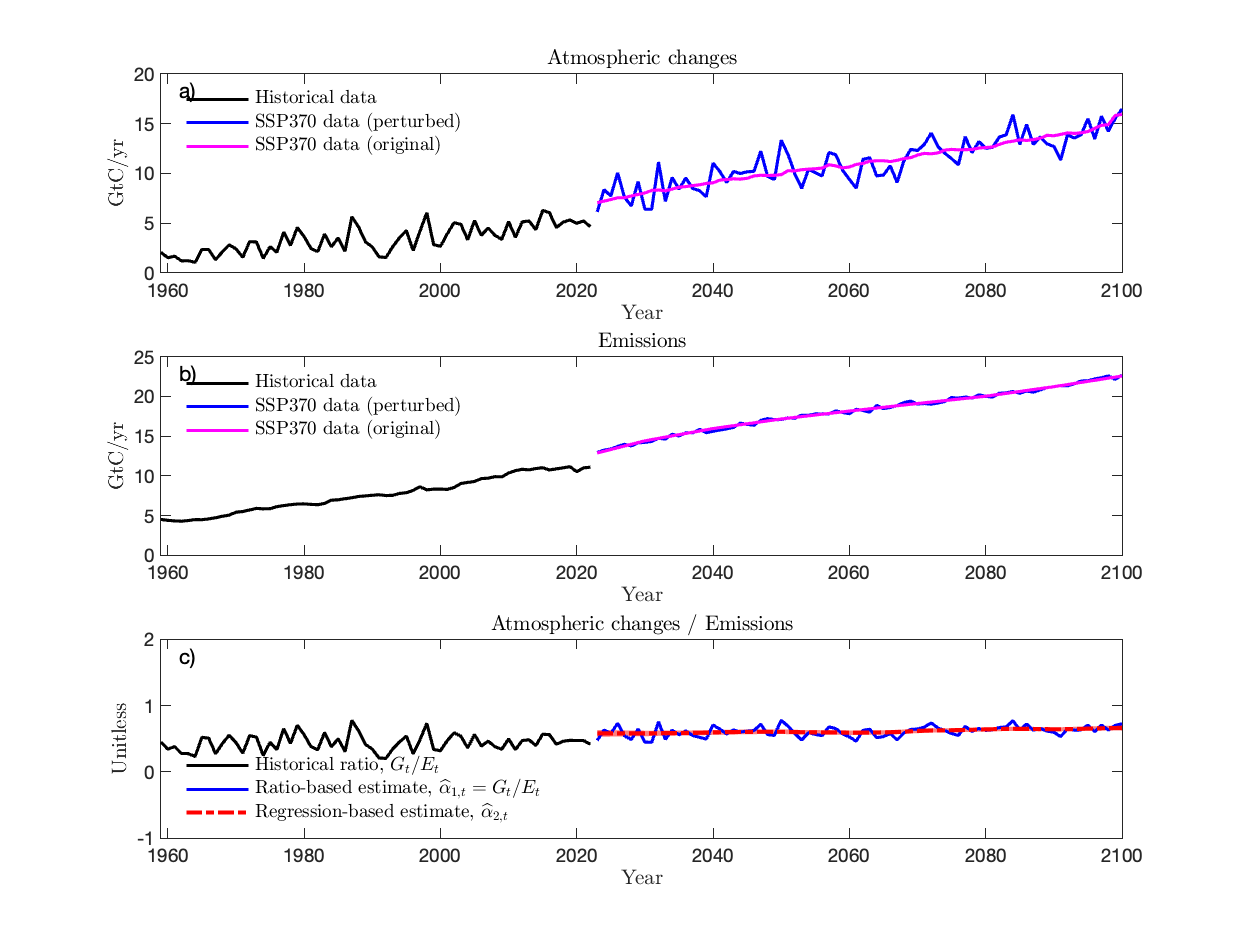}
    \label{fig:SSP370}
\end{figure}

\begin{figure}[t!]
    \centering
   \caption{\emph{Analysis of SSP4-3.4 data over the period $2023$--$2100$. a) Atmospheric concentration changes ($G_t$) for the historical period 1959--2022 (black) and the SSP period 2022--2100.  b) Emissions ($E_t$) data. Magenta lines denote the original output data from MAGICC and the blue lines denote the perturbed data. c) Ratio of atmospheric changes to emissions ($G_t/E_t$). Black lines denote data from the Global Carbon Project over the historical period $1959$--$2022$.  Magenta lines denote the original SSP output data from MAGICC over the period $2023$--$2100$. Blue lines denote the perturbed SSP data. The red line in c) is the regression-based estimator $\hat \alpha_{2,t} = \widehat\E(\alpha_t | \text{data})$ of the time-varying airborne fraction $\alpha_t$, obtained from the Kalman smoother. Shaded area is a $95\%$ confidence band around $\hat \alpha_{2,t}$.} }
    \includegraphics[width=0.99\textwidth]{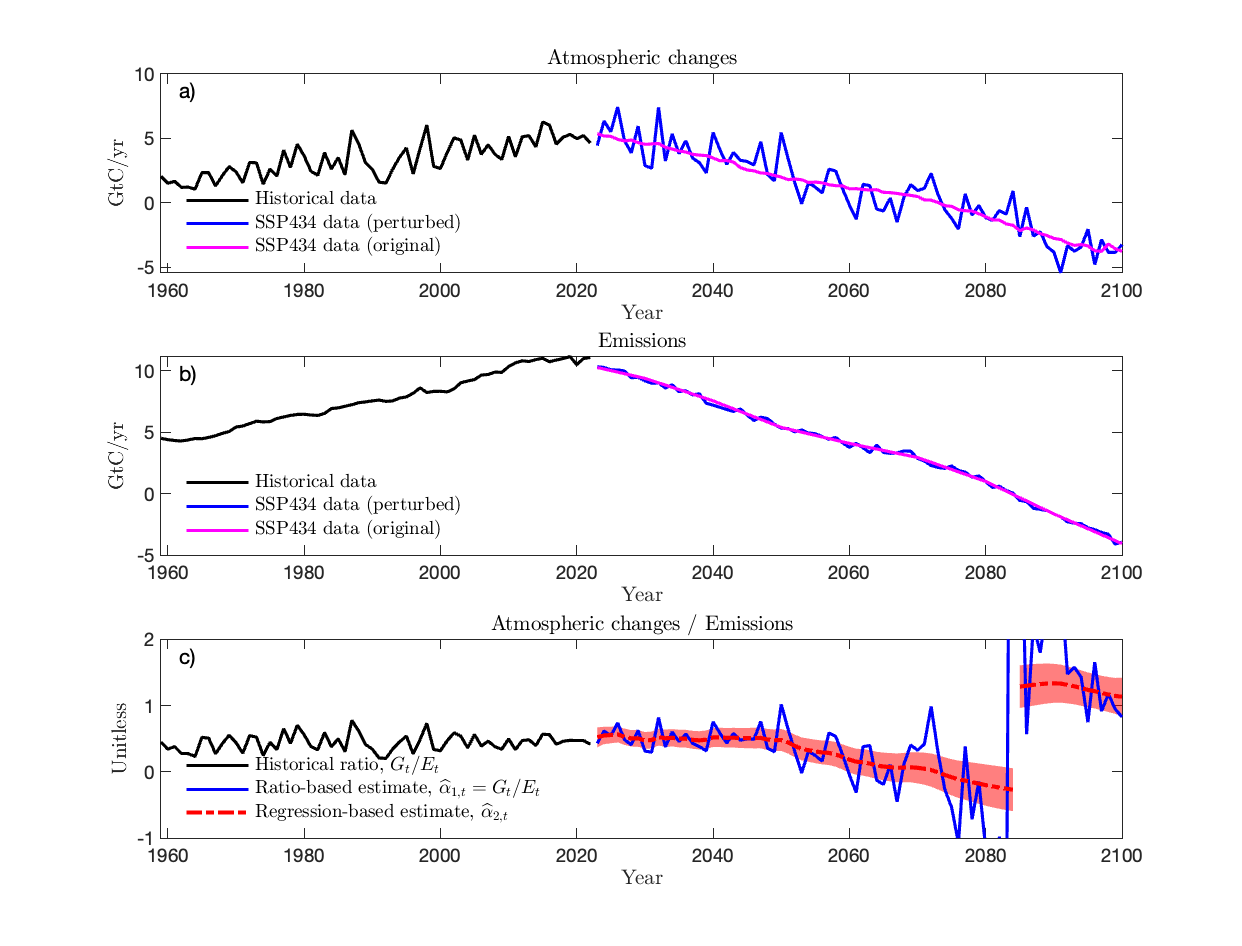}
    \label{fig:SSP434}
\end{figure}

\begin{figure}[t!]
    \centering
   \caption{\emph{Analysis of SSP5-8.5 data over the period $2023$--$2100$. a) Atmospheric concentration changes ($G_t$) for the historical period 1959--2022 (black) and the SSP period 2022--2100.  b) Emissions ($E_t$) data. Magenta lines denote the original output data from MAGICC and the blue lines denote the perturbed data. c) Ratio of atmospheric changes to emissions ($G_t/E_t$). Black lines denote data from the Global Carbon Project over the historical period $1959$--$2022$.  Magenta lines denote the original SSP output data from MAGICC over the period $2023$--$2100$. Blue lines denote the perturbed SSP data. The red line in c) is the regression-based estimator $\hat \alpha_{2,t} = \widehat\E(\alpha_t | \text{data})$ of the time-varying airborne fraction $\alpha_t$, obtained from the Kalman smoother. Shaded area is a $95\%$ confidence band around $\hat \alpha_{2,t}$.} }
    \includegraphics[width=0.99\textwidth]{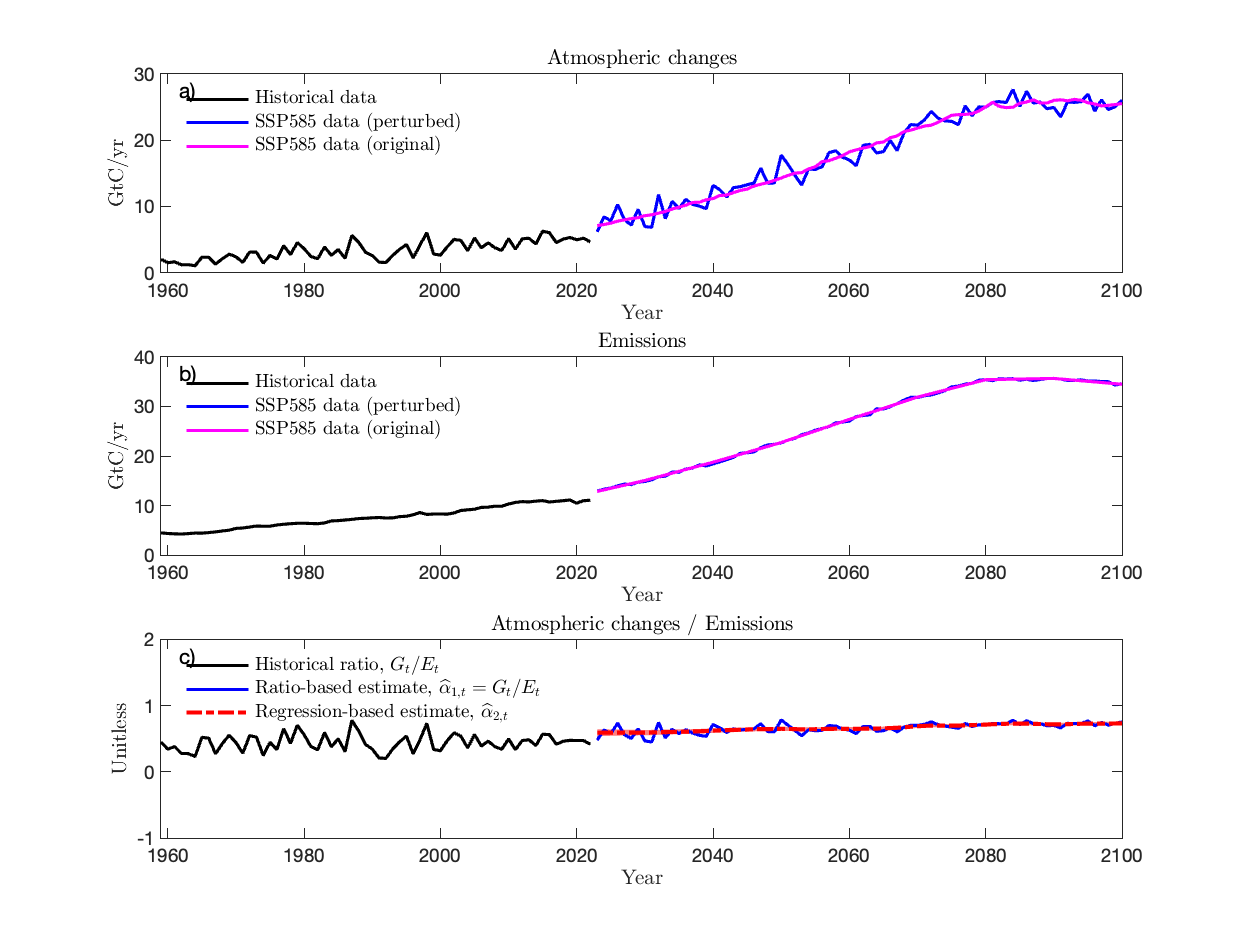}
    \label{fig:SSP585}
\end{figure}

\end{document}